\documentclass[11pt,letter]{article}
\usepackage{natbib}
\usepackage{latexsym}
\usepackage{graphics}
\usepackage{amsfonts}
\usepackage{enumerate}
\usepackage{graphicx}
\usepackage{epsf}
\usepackage{epstopdf}
\usepackage{epsfig}
\usepackage{verbatim}
\usepackage{subfig}
\usepackage{caption}
\usepackage{footnote}
\usepackage{lscape}
\usepackage{hhline}
\usepackage{array}
\usepackage{graphicx}
\usepackage{setspace}
\usepackage{amssymb}
\usepackage{amsmath}
\usepackage{amsfonts}
\usepackage{times}
\usepackage{color}
\usepackage{amsthm}
\usepackage{indentfirst}
\usepackage{epsfig}
\usepackage{lscape,ctable,rotating}
\usepackage{enumitem}
\usepackage{pdfpages}

\setlist{nolistsep}
\usepackage[left=1in,top=1in,right=1in,foot=1in]{geometry}

\usepackage{amsmath}

\DeclareMathOperator*{\argmin}{arg\,min}

\newtheorem{definition}{Definition}[section]

\newtheorem{exam}[definition]{Example}
\newtheorem{proposition}{Proposition}[section]
\newtheorem{rem}{Remark}[section]

\newcommand{\R}{\mathbb{R}}

\renewcommand{\P}{\mathbb{P}}
\newcommand{\Q}{\mathbb{Q}}
\newcommand{\var}{\textrm{\rm var}}
\newcommand{\cov}{\textrm{\rm cov}}

\newcommand{\levy}{{L\'evy}}

\newcommand{\Tau}{{\mathcal{T}}}

\newcommand{\disteq}{\stackrel{\mathrm{d}}{=}}

\newcommand{\subTS}{\textup{\rm subTS}}
\newcommand{\NTS}{\textup{\rm NTS}}
\newcommand{\gNTS}{\textup{\rm gNTS}}
\newcommand{\gStdNTS}{\textup{\rm gStdNTS}}
\newcommand{\stdNTS}{\textup{\rm stdNTS}}
\newcommand{\tr}{{\texttt{T}}}
\newcommand{\diag}{\textrm{\rm diag}}

\newcommand{\sqt}{{\diamond \frac{1}{2}}}
\newcommand\blfootnote[1]{%
  \begingroup
  \renewcommand\thefootnote{}\footnote{#1}%
  \addtocounter{footnote}{-1}%
  \endgroup
}

\title{Quanto Option Pricing on a Multivariate \levy~Process Model with a Generative Artificial Intelligence}
\author{
Young Shin Kim\footnote{College of Business, Stony Brook University, New York, USA (aaron.kim@stonybrook.edu)}
 \and 
 Hyun-Gyoon Kim\footnote{Department of Financial Engineering, School of Business, Ajou University, Korea (hyungyoonkim@ajou.ac.kr).} 
}
\providecommand{\keywords}[1]{\textbf{\textit{Key words:}} #1}

\begin{document}
\maketitle

\begin{abstract}
In this study, we discuss a machine learning technique to price exotic options with two underlying assets based on a non-Gaussian \levy~process model.  We introduce a new multivariate \levy~process model named the generalized normal tempered stable (\gNTS) process,  which is defined by time-changed multivariate Brownian motion. Since the \gNTS~process does not provide a simple analytic formula for the probability density function (PDF), we use the conditional real-valued non-volume preserving (CRealNVP) model, which is a type of flow-based generative network. Then, we discuss the no-arbitrage pricing on the \gNTS~model for pricing the quanto option whose underlying assets consist of a foreign index and foreign exchange rate. We present the training of the CRealNVP model to learn the PDF of the \gNTS~process using a training set generated by Monte Carlo simulation.  Next, we estimate the parameters of the \gNTS~model with the trained CRealNVP model using the empirical data observed in the market.  Finally, we provide a method to find an equivalent martingale measure on the \gNTS~model and to price the quanto option using the CRealNVP model with the risk-neutral parameters of the \gNTS~model.
\blfootnote{\textbf{Acknowledgements:} Y.S. Kim gratefully acknowledges the support of Mr.  Park president of Juro Instruments Co., Ltd., Korea.}

\end{abstract}

\baselineskip=24pt

\doublespacing
\section{Introduction}
A standard quanto option is a European option underlying a foreign asset, whose payoff is converted to another currency at a predefined fixed exchange rate. Since the quanto option provides foreign–asset exposure without taking the corresponding exchange rate risk, the tail dependence between the asset and the exchange rate is instrumental in the valuation. Quanto option pricing based on the Black–Scholes model (\citeauthor{BlackScholes:1973},\citeyear{BlackScholes:1973}), assuming a multivariate Brownian motion, has been studied by \cite{BaxterRennie:1996}. Recently, \cite{KIM2015512} presented the quanto option pricing based on the multivariate normal tempered stable (NTS) process which is a kind of non-Gaussian \levy~process. This approach is more efficient than the Gaussian approach since the NTS process can capture the fat-tails and asymmetric dependence between the asset and the exchange rate, which are empirically observed in the market. The NTS process model is more realistic than the multivariate Brownian motion, but it still has restrictions. 
The NTS process is defined by taking multivariate Brownian motion and substituting the Tempered Stable Subordinator with the time variable.
In this definition, only one subordinator is applied to different elements of the multivariate Brownian motion.
Since the subordinator is related to the time-varying volatility of the market, the NTS model supposes that only one market volatility affects various assets in the market.
However, the volatility characteristics of a foreign asset and of a foreign exchange rate are different, and the single subordinator setting of the NTS model cannot explain this difference, and hence it is not realistic to model the quanto option pricing.

In this research, we provide a generalized NTS (gNTS) process which is defined by a mixture of multiple subordinators to multivariate Brownian motion. This enhanced process not only captures fat-tails and asymmetric dependence of multi-dimensional asset returns but also describes the different volatility characteristics of a foreign asset and exchange rate. As a consequence, we are allowed to obtain a more flexible quanto option pricing model by the \gNTS~process. Moreover, the \gNTS~model allows us to find risk-neutral measures using Sato's change of measure in \levy~process model (\citeauthor{Sato:1999}, \citeyear{Sato:1999}) and Girsanov's theorem. We find option prices under the \gNTS~model based on the risk-neutral parameters for the risk-neutral measure equivalent to the physical market measure fitted to the empirical data.

Since the probability density function (PDF) of the gNTS process is not given by a simple analytic form, we need to have an efficient numerical method to apply the model to derivative pricing such as Quanto options. The Monte Carlo method can be a good alternative, but the simulation takes a long time and is not easy to obtain sensitivity, such as the Greek Letters of the option. In this paper, we suggest an extension of the real-valued non-volume preserving (RealNVP) model to obtain the PDF of the \gNTS~process.
First, we demonstrate flow-based generative networks based on the RealNVP designed by \cite{Dinh_et_al:2016}. As other generative models including Generative Adversarial Network (\citeauthor{Goodfellow_et_al:2014}, \citeyear{Goodfellow_et_al:2014}) and Variational Autoencoder (\citeauthor{Kingma_el_al:2013}, \citeyear{Kingma_el_al:2013}), this generative model can learn the probability density inherent in data and generate new data samples that resemble the original data. Furthermore, only flow-based generative models are able to provide the density functions in explicit form while other generative networks cannot. Since the original form of the RealNVP model is nonparametric, it has difficulty in the arbitrage option pricing theory, which needs to find the risk-neutral measure. To overcome this drawback, we use the Conditional RealNVP (CRealNVP) model by \cite{KimKwonKimHuh:2022}. The CRealNVP allows model parameters of a given parametric distribution as input variables. In the option pricing with the \gNTS~model, we will find a set of risk-neural parameters of the risk-neutral measure of the physical market measure. The CRealNVP can be applied to find the PDF of the \gNTS~process, to estimate the parameters of the \gNTS~market model, and to calculate the Quanto option pricing under the \gNTS~model with the risk-neutral parameters.

The remainder of this paper is organized as follows. The review of the NTS process is presented in Section 2.
Section 3 proposes how we construct the gNTS process and standard gNTS process. Section 4 presents the 2-dimensional gNTS model for an underlying asset return and a foreign exchange rate return, and discusses the change of measures between the physical and risk-neutral measures on the model. 
In section 5, we demonstrate the CRealNVP model: definition of the model, training the CRealNVP model for the gNTS model with a training set generated by Monte-Carlo simulation, and \gNTS~model parameter estimation using the historical data through the trained CRealNVP model. In addition, we provide a method to select a set of risk-neural parameters using the estimated physical market parameters. A calculating method for the quanto option price using the CRealNVP model under the risk-neutral parameters of the \gNTS~model is also proposed in this section. Section 6 concludes followed by the proofs and mathematical details in the Appendix.

\section{NTS Processes}
Let $\alpha \in (0, 2)$, $\theta > 0$, and $c>0$. Assume \levy~measure $\nu$ equals to
\[
\nu(dx) = \frac{-ce^{-\theta x}}{\Gamma\left(-\frac{\alpha}{2}\right)x^{\alpha/2+1}} 1_{x>0}dx
\]
and let $\gamma = \int_0^1x\nu(dx)$. A pure jump \levy~process $\mathcal T = (\Tau(t))_{t \geq 0}$ defined by the \levy-Khintchine formula
\[
\phi_{\Tau(t)}(u)=\exp\left(i\gamma u t+t\int_{-\infty}^\infty(e^{iux}-1-iux1_{|x|\le1})\nu(dx)\right)
\]
is referred to as the \textit{tempered stable subordinator} with parameters $(\alpha$, $c$, $\theta)$ and denoted to $\Tau\sim \subTS(\alpha, c, \theta)$. 
The characteristic function $\phi_{\Tau(t)}(u)$ of the  \levy-Khintchine formula is simplified to
\begin{equation}\label{eq:chfTSSub}
\phi_{\Tau(t)}(u)=\exp\left(-ct\left(\left(\theta-iu\right)^{\frac{\alpha}{2}}-\theta^{\frac{\alpha}{2}}\right)\right).
\end{equation}

By applying Sato's change of measure theorem (\citeauthor{Sato:1999},\citeyear{Sato:1999}), we can prove the following proposition\footnote{See \cite{KimLee:2007} and \cite{Kim:2005} for details.}.
\begin{proposition}\label{pro:Change of Measure Tau}
Consider a measure $\P$ and assume that $\Tau\sim \subTS(\alpha, c, \theta)$ under $\P$.
Then there is a measure $\Q_{\hat\theta}$ equivalent to $\P$ such that $\Tau\sim \subTS(\alpha, c, \hat\theta)$ under $\Q_{\hat\theta}$ for $\hat\theta>0$.
\end{proposition}

Let $N$ be a positive integer, and $\R_+ = \{x\in\R| x>0\}$ be the set of positive real numbers.
Consider a subordinator $\Tau=(\Tau(t))_{t \geq 0}$ $\sim$ $\subTS(\alpha$, $c$, $\theta)$ where $\alpha \in (0,2)$, $\theta > 0$, and $c = \frac{2\theta^{1-\alpha/2}}{\alpha}$.
Let $\mu \in\R^N$, $\beta \in\R^N$, and $\sigma \in\R_+^N$, where $\mu_n$, $\beta_n$ and $\sigma_n$ are considered as the $n$-th elements of $\mu$, $\beta$, and $\sigma$, respectively.  
Let $R = \left[\rho_{k,n}\right]_{k, n \in \{1, 2, ..., N\}}$ be a dispersion matrix with $\rho_{n,n}=1$ and $ R^{1/2}$ given by factorization $R=R^{1/2}(R^{1/2})^\tr$, such as a Cholesky factorization. 
Assume that $B = (B(t))_{t \geq 0}$ 
 is an independent $N$-dimensional Brownian motion and $B$ is independent of $\Tau$. 
The $N$-dimensional process $X = (X(t))_{t \geq 0}$, defined by 
\[
	X(t) = \mu t + \beta \Tau(t)+ \textup{diag}(\sigma) R^{1/2}B(\Tau(t))
\]
is called an $N$-dimensional NTS process and denoted by\footnote{The parameters are as follows: $\alpha\in(0,2]$ determines fat-tailedness and peakedness (with smaller values implying fatter tails and higher peaks) as well as the jump intensity, implying infinite variation for $\alpha\in[1,2)$ and finite variation for $\alpha\in(0,1)$; $\theta$ is the tempering and scaling parameter for the subordinator; $\mu$ reflects the drift of the NTS process; $\beta$ and $\sigma$ are the skewness and scale parameters, respectively; and $R$ determines the dependence structure.}
\[
	X \sim \NTS_N\left(\alpha, \theta, \beta, \mu, \sigma, R\right).
\]
The characteristic function of $X_n(t)$, the $n$-th element of $X(t)$, is
\[
	\phi_{X_n(t)}(u) = \exp\left(i \mu_n u t - \frac{2t\theta^{1-\frac{\alpha}{2}}}{\alpha}\left(\left(\theta - \beta_n i u + \frac{\sigma_n^2 u^2}2\right)^{\frac\alpha 2} - \theta^{\frac\alpha 2}\right) \right), 
\]
for $n\in\{1,2,\cdots, N\}$.
The expectation and the covariance are $E[X_n(t)] = \left(\mu_n+\beta_n\right)t$ and
\begin{equation}\label{cov}
 \cov(X_k(t), X_n(t)) = \sigma_k\sigma_n \rho_{k, n} t + \beta_k \beta_n t\left(\frac{2 - \alpha}{2 \theta}\right), 
\end{equation}
respectively, for $k,n\in\{1,2,\cdots, N\}$.

If we set $\mu=-\beta$ and $\sigma_n = \sqrt{1-\beta_n^2\left(\frac{2-\alpha}{2\theta}\right)}$ with $|\beta_n|<\sqrt{\frac{2\theta}{2-\alpha}}$ for $n\in\{1,2,\cdots, N\}$, then $X_0\sim \textup{NTS}(\alpha$, $\theta$, $\beta$, $\sigma$, $\mu$, $R)$ has $E[X_0(t)]=0$ and $\var(X_0(t))=t(1,1,\cdots, 1)^\tr$. In this case, we say that $X_0$ is the \textit{standard NTS process} and denote $X_0\sim \textup{stdNTS}(\alpha$, $\theta$, $\beta$, $R)$. A process new $X=(X(t))_{t\ge0}$ defined as $X(t)=mt+\diag(s)X_0(t)$ for $m\in\R^N$ and $s\in\R_+^N$ becomes
\[
X\sim \NTS_N(\alpha, \theta, \diag(s)\beta, m-\beta, \sigma, R)
\]
where $\sigma=(\sigma_1, \sigma_2, \cdots, \sigma_N)^\tr\in\R_+^N$ of which the $n$-th element is $\sigma_n=s_n\sqrt{1-\beta_n^2\left(\frac{2-\alpha}{2\theta}\right)}$ and $s_n$ is the the $n$-th element of $s$.
More details of the multivariate NTS distribution and process can be found in literature including \cite{kim2023multi}, \cite{Kim:2022}, \cite{KurosakiKim:2018}, \cite{Anad_et_al:2016}, and \cite{KimVolkmann:2013}. \cite{KIM2015512} presented the quanto option pricing under the 2-dimensional NTS process model.

\section{Generalized NTS Processes} \label{sec:NdimensionalNTSProcess}
Let $N$ be a positive integer, and $I_2=(0,2)$ be an open interval between 0 and 2. 
We consider an $N$-dimensional vectors $\alpha \in I_2^N$,  $\theta \in \R_+^N$, $\beta\in\R^N$, and $\sigma\in \R_+^N$,
where $\alpha_n$, $\theta_n$, $\beta_n$ and $\sigma_n$ are the $n$-th elements of $\alpha$, $\theta$, $\beta$, and $\sigma$, respectively.
Let $R = \left[\rho_{k,n}\right]_{k, n \in \{1, 2, ..., N\}}$ be a dispersion matrix with $\rho_{n,n}=1$.
Let $\Tau = (\Tau(t))_{t \geq 0}$ be a $N$-dimensional independent tempered stable subordinator with $\Tau(t) = (\Tau_1(t), \Tau_2(t), \cdots, \Tau_N(t))^\tr$ and $(\Tau_n(t))_{t\ge 0}\sim \subTS(\alpha_n, 1, \theta_n)$ for $n\in\{1,2\cdots, N\}$\footnote{To simplify the model, we set $c=1$ in the tempered stable subordinator.}.
Let $B = (B(t))_{t \geq 0}$ be an independent $N$-dimensional Brownian motion 
and assume $B$ and $\Tau$ are all mutually independent.
Suppose there is a $N$-dimensional process $(\tau(t))_{t\ge0}$ with $\tau(t) = (\tau_1(t),\tau_2(t), \cdots, \tau_N(t))^\tr$ satisfying $\Tau_n(t)=\int_0^t\tau_n(u)du$, for all $t\ge0$ and for $n\in\{1,2,\cdots, N\}$. Let $\tau^\sqt(t)=\left(\sqrt{\tau_1(t)},\sqrt{\tau_2(t)}, \cdots, \sqrt{\tau_N(t)}\right)^\tr$.
The $N$-dimensional process $X = (X(t))_{t \geq 0}$ defined by
\[
X(t) = \mu t+\diag\left(\beta\right) \int_0^t\tau(u)du +  \textup{diag}(\sigma)  \int_0^t \diag\left(\tau^\sqt(t)\right)R^{\frac{1}{2}} dB(u) 
\]
is called an $N$-dimensional generalized NTS process and is denoted by
\[
	X \sim \gNTS_N\left(\alpha, \theta, \beta, \mu, \sigma, R\right).
\]

Let $B_n^0(t)$ be the $n$-th element of $R^{\frac{1}{2}}B(t)$ for $t\ge 0$. Then the process $B_n^0 = (B_n^0(t))_{t\ge0}$ is a Brownian motion 
and $X_n(t)$, the $n$-th element of $X(t)$, is given by 
\begin{align}
\nonumber
X_n(t) &= \mu_n t + \beta_n \int_0^t \tau_n(u) du + \sigma_n \int_0^t \sqrt{\tau_n(u)}dB_n^0(u) \\
\label{eq:timechangedBM}
&= \mu_n t + \beta_n \Tau_n(t) + \sigma_nB_n^0(\Tau_n(t)).
\end{align}
Note that, we have $dB_k^0(t)\cdot dB_n^0(t) = \rho_{k,n}\,dt$.

\begin{proposition}\label{pro:GirsanovGNTS}
Suppose $X \sim \gNTS_N\left(\alpha, \theta, \beta, \mu, \sigma, R\right)$ under measure $\P$. 
Then there is an equivalent measure $\Q_{\hat\theta, \hat\beta}$ for $\hat\theta\in\R^+$ and $\hat\beta\in\R$, and 
\[
X \sim \gNTS_N\left(\alpha, \hat\theta, \hat\beta, \mu, \sigma, R\right)
\]
under the measure $\Q_{\hat\theta, \hat\beta}$.
\end{proposition}

Let $X\sim\gNTS_N(\alpha, \theta, \beta, \mu, \sigma, R)$. Then the characteristic function of $X_n(t)$, the $n$-th element of $X(t)$, is
\[
	\phi_{X_n(t)}(u) = \exp\left(\mu_n i u t - t\left(\left(\theta_n - \beta_n i u + \frac{\sigma_n^2 u^2}2\right)^{\frac{\alpha_n}{2}} - \theta_n^{\frac{\alpha_n}{2}}\right) \right),
\]
for $n\in\{1,2,\cdots, N\}$.
Using the first and second derivatives of $\phi_{X_n(t)}$, we obtain the expectation as $E[X_n(t)] = \left(\mu_n + \frac{1}{2}\alpha_n\beta_n \theta_n^{\frac{\alpha_n}{2}-1}\right)t$ and the variance as
\begin{equation}\label{eq:variance}
\var(X_n(t)) = \frac{\alpha_n\theta_n^{\frac{\alpha}{2}-1}}{2}\left(\left(\frac{2-\alpha_n}{2\theta_n}\right)\beta_n^2+\sigma_n^2\right)t,
\end{equation}
for $n\in\{1,2,\cdots, N\}$.

Suppose that
\[
-\frac{2\theta_n^{1-\frac{\alpha_n}{4}}}{\sqrt{\alpha_n(2-\alpha_n)}}<\beta_n<\frac{2\theta_n^{1-\frac{\alpha_n}{4}}}{\sqrt{\alpha_n(2-\alpha_n)}}
~~~\text{ for } ~~~ n\in\{1,2,\cdots, N\},
\]
and define two vectors $\mu_0\in\R^N$ and  $\sigma_0\in\R_+^N$ 
where the $n$-th elements of $\mu_0$ and $\sigma_0$ are given by 
\begin{equation}\label{eq:stdparam of mu and sigma}
\mu_{0,n}=-\frac{1}{2}\alpha_n \beta_n \theta_n^{\frac{\alpha_n}{2}-1}, \text{ and }
\sigma_{0,n}=\sqrt{\frac{2}{\alpha_n}\theta_n^{1-\frac{\alpha_n}{2}}-\frac{2-\alpha_n}{2\theta_n}\beta_n^2}
\end{equation}
for $n\in\{1,2,\cdots, N\}$, respectively.
Then a \gNTS~process $X_0\sim\gNTS_N(\alpha, \theta, \beta, \mu_0, \sigma_0, R)$ has properties
$E[X_0(t)] = (0,0,\cdots, 0)^\tr$ and $\var(X_0(t)) = t(1,1,\cdots, 1)^\tr$.
In this case, the process $X_0$ is referred to as the \textit{standard gNTS} process with parameters $(\alpha,\theta,\beta,R)$ and denoted as
\[
X_0\sim \gStdNTS_N(\alpha,\theta,\beta,R).
\]
Using the standard \gNTS~process, we obtain the following proposition without proof:
\begin{proposition}\label{prop:stdgnts to gnts}
(a) Suppose $X_0\sim \gStdNTS_N(\alpha,\theta,\beta,R)$. Then a new process $X=(X(t))_{t\ge 0}$ with 
\[
X(t) = m t + \diag(s) X_0(t)
\]
for $m\in\R^N$ and $s\in\R_+^N$ becomes 
\[
X\sim\gNTS_N(\alpha, \theta, \diag(s)\beta , \diag(s)\mu_0+m, \diag(s)\sigma_0, R),
\]
and $E[X(t)] = m t$ and $\var(X(t)) = s^2 t$.\\
(b) Conversely, suppose $X\sim\gNTS_N(\alpha, \theta, \beta, \mu, \sigma, R)$. Then $X$ can be represented by the standard \gNTS~process as
\[
X(t) = mt+\diag(s) X_0(t)
\]
with $X_0\sim \gStdNTS_N(\alpha,\theta,\bar\beta,R)$, 
where $m\in\R^N$, $s\in\R_+^N$, and $\bar\beta\in\R^N$ of which the $n$-th elements are
\begin{align*}
m_n =\mu_n+\frac{\alpha_n\beta_n}{2}\theta_n^{\frac{\alpha_n}{2}-1}
,
s_n = \sqrt{\frac{\alpha_n}{2}\theta_n^{\frac{\alpha_n}{2}-1}\left(\left(\frac{2-\alpha_n}{2\theta_n}\right)\beta_n^2+\sigma_n^2\right)}
,
\text{ and }
\bar\beta_n = \frac{\beta_n}{s_n},
\end{align*}
respectively. Here, $\alpha_n$, $\theta_n$, $\beta_n$, $\mu_n$, and $\sigma_n$ are the $n$-th elements of $\alpha$, $\theta$, $\beta$, $\mu$, and $\sigma$, respectively.
\end{proposition}

Suppose that $(X(t))_{t\ge0}$ is an arithmetic Brownian motion given by 
\[
X(t) = \mu t + \diag(\sigma) R^{\frac{1}{2}}B(t).
\]
Then we know that
\[
X(t)\disteq \mu t + \sqrt{t}\diag(\sigma) R^{\frac{1}{2}}B(1).
\]
If we replace a symmetric $\alpha$-stable process $(L(t))_{t\ge0}$ instead of $(B(t))_{t\ge0}$,
then we have
\[
X(t)\disteq \mu t + t^{1/\alpha} \diag(\sigma)R^{\frac{1}{2}} L(1),
\]
since $L(t)\disteq t^{1/\alpha}L(1)$. 
Applying the same arguments to the \gNTS~process case, we obtain a non-trivial result as the following proposition.
\begin{proposition}\label{pro:gNTS normalization in time}
Suppose $T>0$ and $X\sim\gNTS_N(\alpha, \theta, \beta, \mu, \sigma, R).$
Then we have 
\[
X(T) \disteq m+\diag(s) \Xi(1)
~~~\text{ for }~~~\Xi\sim \gStdNTS_N(\alpha,\theta_\Xi,\beta_\Xi,R),
\]
where the $n$-th elements of $\theta_\Xi$, $m\in\R^N$, $s\in\R_+^N$, and $\beta_\Xi\in\R^N$ are
\begin{align*}
m_n =T\left(\mu_n+\frac{\alpha_n\beta_n}{2}\theta_n^{\frac{\alpha_n}{2}-1}\right)
,
s_n = \sqrt{\frac{\alpha_n}{2}\theta_n^{\frac{\alpha_n}{2}-1}T\left(\left(\frac{2-\alpha_n}{2\theta_n}\right)\beta_n^2+\sigma_n^2\right)}
,
\theta_{\Xi,n}=\theta_n T^{\frac{2}{\alpha_n}}
,
\end{align*}
and $\beta_{\Xi,n}=\frac{\beta_n T^{\frac{2}{\alpha_n}}}{s_n}$, respectively. 
Here, $\alpha_n$, $\theta_n$, $\beta_n$, $\mu_n$, and $\sigma_n$ are the $n$-th elements of $\alpha$, $\theta$, $\beta$, $\mu$, and $\sigma$, respectively.
\end{proposition}

\section{\label{sec:QuantoOptionPricingOngNTS}Quanto Option Pricing on gNTS Model}
We denote the domestic and the foreign risk-free interest rates by $r_d$ and $r_f$, respectively. 
Then, let $(S(t))_{t \geq 0}$ be the price process for the asset in foreign currency, 
$(V(t))_{t \geq 0}$ be the price process of the asset in domestic currency, 
and $(F(t))_{t \geq 0}$ be the foreign exchange (FX) rate process of the foreign currency with respect to the domestic currency.
That means $V(t) = F(t)S(t)$. 
We assume that $(F(t))_{t\ge0}$ and $(V(t))_{t\ge0}$ are given by
\begin{equation}\label{eq:F&V in P}
\begin{cases}
F(t) = F(0)\exp(X_F(t))\\
V(t) = V(0)\exp(X_V(t))
\end{cases} \text{ for } t\ge0,
\end{equation}
where $X(t) = (X_F(t), X_V(t))^\tr$ and 
\begin{equation}\label{eq:Y in P}
(X(t))_{t\ge0}
\sim \gNTS_2\left(
\left(\begin{matrix} \alpha_F\\ \alpha_V\end{matrix}\right),
\left(\begin{matrix} \theta_F\\ \theta_V\end{matrix}\right),
\left(\begin{matrix} \beta_F\\ \beta_V\end{matrix}\right),
\left(\begin{matrix} \mu_F\\ \mu_V\end{matrix}\right),
\left(\begin{matrix} \sigma_F\\ \sigma_V\end{matrix}\right),
\left(\begin{matrix} 1 &  \rho \\ \rho & 1\end{matrix}\right)
\right)
\end{equation}
under the physical (or market) measure $\P$. Then by Proposition \ref{pro:GirsanovGNTS}, we can find equivalent measure $\Q_{\hat\theta, \hat\beta}$ under which
\[
(X(t))_{t\ge0}
\sim \gNTS_2\left(
\left(\begin{matrix} \alpha_F\\ \alpha_V\end{matrix}\right),
\left(\begin{matrix} \hat\theta_F\\ \hat\theta_V\end{matrix}\right),
\left(\begin{matrix} \hat\beta_F\\ \hat\beta_V\end{matrix}\right),
\left(\begin{matrix} \mu_F\\ \mu_V\end{matrix}\right),
\left(\begin{matrix} \sigma_F\\ \sigma_V\end{matrix}\right),
\left(\begin{matrix} 1 &  \rho \\ \rho & 1\end{matrix}\right)
\right).
\]

To derive the risk-neutral measure, we have to find an equivalent measure, $\Q_{\hat\theta^*, \hat\beta^*}$, with $\hat\theta^*=(\hat\theta_F^*,\hat\theta_V^*)^\tr$ and $\hat\beta^*=(\hat\beta_F^*, \hat\beta_V^*)^\tr$, 
under which the discounted price processes $(\tilde F(t))_{t \geq 0}$ and $(\tilde V(t))_{t \geq 0}$ are martingales,
where $\tilde F(t)$ $=e^{(-r_d+r_f)t}F(t)$ and $\tilde V(t)=e^{-r_dt}V(t)$. The martingale property is satisfied if 
\[
E_{\Q_{\hat\theta^*, \hat\beta^*}}\left[\tilde F(t)\right]=F(0)
~~~\text{ and }~~~ E_{\Q_{\hat\theta^*, \hat\beta^*}}\left[\tilde V(t)\right]=V(0),
\]
which are equivalent to
\[
E_{\Q_{\hat\theta^*, \hat\beta^*}}\left[e^{X_F(t)}\right]=e^{(r_d-r_f)t}
~~~\text{ and }~~~
E_{\Q_{\hat\theta^*, \hat\beta^*}}\left[ e^{X_V(t)}\right]=e^{r_d t}.
\] 
These conditions are equivalent to
$\log \phi_{X_F(1)}(-i)=r_d-r_f$ and $\log \phi_{X_V(1)}(-i)=r_d$, respectively. That is
\[
\mu_F-\left(\left(\hat\theta^*_F-\hat\beta^*_F-\frac{\sigma_F^2}{2}\right)^{\frac{\alpha_F}{2}}-\left(\hat\theta^*_F\right)^\frac{\alpha_F}{2}\right)=r_d-r_f
\]
and
\[
\mu_V-\left(\left(\hat\theta^*_V-\hat\beta_V^*-\frac{\sigma_V^2}{2}\right)^{\frac{\alpha_V}{2}}-\left(\hat\theta^*_V\right)^\frac{\alpha_V}{2}\right)=r_d.
\]
Hence, $\hat\theta^*$ and $\hat\beta^*$ must satisfy:
\begin{description}
\item[RN.1:] $\hat\theta^*_F - \hat\beta^*_F - \frac{\sigma_F^2}2>0$ and $\theta^*_V - \hat\beta_V^* - \frac{\sigma_V^2}2>0$ for $E_{\Q_{\hat\theta^*, \hat\beta^*}}\left[ e^{X_F(t)}\right]$ and $E_{\Q_{\hat\theta^*, \hat\beta^*}}\left[ e^{X_V(t)}\right]$ to exist.
\item[RN.2:] The discounted price processes $(\tilde F(t))_{t \geq 0}$ and $(\tilde V(t))_{t \geq 0}$ are martingales, which are equivalent to
\[
\mu_F=r_d-r_f+\left(\hat\theta^*_F-\hat\beta^*_F-\frac{\sigma_F^2}{2}\right)^{\frac{\alpha_F}{2}}-\left(\hat\theta^*_F\right)^\frac{\alpha_F}{2}
\]
and
\[
\mu_V=r_d+\left(\hat\theta^*_V-\hat\beta_V^*-\frac{\sigma_V^2}{2}\right)^{\frac{\alpha_V}{2}}-\left(\hat\theta^*_V\right)^\frac{\alpha_V}{2}.
\]
\end{description}

We have the quanto call option payoff function $F_\text{fix}(S(T) - K)^+$ with the time to maturity $T$, strike price $K$ and the fixed exchange rate $F_\text{fix}$, where
$S(T)=\frac{V(T)}{F(T)}$. By \eqref{eq:F&V in P}, we have $S(T)=S(0)\exp(X_V(T)-X_F(T))$.
Therefore, the current option price is obtained by
\begin{equation}\label{eq:QuantoOptionPrice Expectation form}
E_{\Q_{\hat\theta^*, \hat\beta^*}}\left[e^{-r_d T}F_\text{fix}(S(T) - K)^+\right]
=e^{-r_d T}F_\text{fix}E_{\Q_{\hat\theta^*, \hat\beta^*}}\left[(S(0)\exp(X_F(T)-X_V(T)) - K)^+\right].
\end{equation}

\section{Conditional Real NVP}
Let $j\in \{1,2,\cdots, J\}$ where $J$ is the number of coupling layers. Define a $N$-dimensional masking vector $b$ as 
$
b = (\underbrace{1, \cdots, 1}_{n\rm\ times}, \underbrace{0, \cdots, 0}_{N-n\rm\ times})^\tr.
$
We set a sequence of the masking vectors as $b^{(1)}=b$ and $b^{(j+1)}=I-b^{(j)}$ where a $N$-dimensional unit vector $I=(1,1,\cdots,1)^\tr$.

Let $y$ be a given $N$-dimensional column vector and define a affine coupling layer function $f^{(j)}:\R^N \rightarrow \R^N$ as
\[
f^{(j)}(y)=b^{(j)}\odot y+
\left(I-b^{(j)}
\right)\odot
\left(
	\left( y - \textbf{t}^{(j)}
		\left(
			b^{(j)}\odot y
		\right)
	\right)\odot\exp
	\left(
		-\textbf{s}^{(j)}( b^{(j)}\odot y )
	\right)
\right),
\]
where the scale function $\textbf{s}^{(j)}$ and translation function $\textbf{t}^{(j)}$ are both functions from $\R^N$ to $\R^N$, respectively, the function $\exp(\cdot)$ is the element-wise exponential function, and $\odot$ is an element-wise product. Here, the functions $\textbf{s}^{(j)}$ and $\textbf{t}^{(j)}$ are represented by deep-neural-networks. The flow-based generative neural network $f$ composed of $f = f^{(J)}\circ f^{(J-1)}\circ \cdots \circ f^{(1)}$ is called the \textit{real-valued non-volume preserving transformation model} or simply the \textit{RealNVP model}.
Consider a random variable $Z$ with a PDF $p_Z$,  and a random variable $Y$ with a PDF $p_Y$. 
We assume that $Z=f(Y)$. By the change of variables, the relation between $p_Y$ and $p_Z$  is given as
\begin{align*}
p_Y (y) 
&= p_Z \left(f(y)\right)\left|\det\left(\frac{\partial f(y)}{\partial y}\right)\right|
\\
&=p_Z(f(y))\prod_{j=1}^J\exp\left(-\left(I-b^{(j)}\right)^\tr\cdot \textbf{s}^{(j)}\left(b^{(j)}\odot y^{(j)}\right)\right),
\end{align*}
where $y^{(1)}=y$ and $y^{(j)}=f^{(j-1)}\circ f^{(j-2)}\circ\cdots\circ f^{(1)}(y)$ for $j>1$. For simplicity, we choose the multivariate standard Gaussian distribution for the prior distribution $p_Z$.

In order to apply the RealNVP transformations to \gNTS~model, we consider a set of model parameters $\Theta$.
The function $f_\Theta^{(j)}$ for all $j$-th affine coupling layers is defined as follows:
\[
f_\Theta^{(j)}(y)
=b^{(j)}\odot y+\left(I-b^{(j)}\right)
\odot\left(\left( y - \textbf{t}^{(j)}\left(b^{(j)}\odot y; \Theta\right)\right)\odot\exp\left(-\textbf{s}^{(j)}\left( b^{(j)}\odot y; \Theta \right)\right)\right),
\]
where $\textbf{s}^{(j)}$ and $\textbf{t}^{(j)}$ are represented by deep neural networks whose input variables consist of the $N$-dimensional $b^{(j)}\odot x$ and the set of parameters $\Theta$. 
We define the conditional flow-based function $f_\Theta$ composed by $f_\Theta = f_\Theta^{(J)}\circ f_\Theta^{(J-1)}\circ \cdots \circ f_\Theta^{(1)}$. Consider a random variable $Y_\Theta$ with a PDF $p_{Y_\Theta}$ and $Z$  with a PDF $p_Z$. We assume that $Z=f_\Theta(Y_\Theta)$. Then the PDF of $Y_\Theta$ under $\Theta$ is obtained using $p_Z$ as
\begin{align*}
p_{Y_\Theta} (y)
&=p_Z\left(f_\Theta(y)\right)\prod_{j=1}^J\exp\left(-\left(I-b^{(j)}\right)^\tr\cdot \textbf{s}^{(j)}\left(b^{(j)}\odot y^{(j)};\Theta\right)\right),
\end{align*}
where $y^{(1)}=y$ and $y^{(j)}=f_\Theta^{(j-1)}\circ f_\Theta^{(j-2)}\circ\cdots\circ f_\Theta^{(1)}(y)$ for $j>1$.
Note that if the neural networks $\textbf{s}^{(j)}$ and $\textbf{t}^{(j)}$ are trained to allow $z=f_\Theta (y)$ to follow the prior standard Gaussian distribution regardless of $\Theta$, the PDF of $Y$ can be explicitly estimated. 
In this case, this generalized RealNVP model is referred to as the \textit{conditional RealNVP (CRealNVP) model}.

\subsection{Training CRealNVP for 2-Dimensional \gStdNTS~ Distribution}
We take 2-dimensional \gStdNTS~model:
\[
\Xi 
\sim \gStdNTS_2\left(
\left(\begin{matrix} \alpha_1\\ \alpha_2\end{matrix}\right),
\left(\begin{matrix} \theta_1\\ \theta_2 \end{matrix}\right),
\left(\begin{matrix} \beta_1 \\ \beta_2 \end{matrix}\right),
\left(\begin{matrix} 1 &  \rho \\ \rho & 1\end{matrix}\right)
\right)
\]
We generate a set of gStdNTS parameters $\alpha_1$, $\alpha_2$, $\theta_1$, $\theta_2$, $\beta_1$, $\beta_2$ and $\rho$ randomly as follows:
\[
\alpha_n=2U_{1,n},~~~ \theta_n=10\tan\left(\frac{\pi U_{2,n}}{2}\right),~~~ \beta_n=\frac{2\theta_n\left(1-\frac{\alpha_n}{4}\right)}{\sqrt{\alpha_n(2-\alpha_n)}}(2U_{3,1}-1),~ \text{ and }~\rho=2U_{4,n}-1,
\]
where $U_{l,n}\sim \text{Beta}(2,2)$, that is a Beta distributed random number with parameters (2,2), for $l\in\{1,2,3,4\}$ and $n\in\{1,2\}$.
Then we generate $2^{10}$ number of \gStdNTS~random vectors of $\Xi(1)$ using the equation \eqref{eq:timechangedBM} with standard parameters given in \eqref{eq:stdparam of mu and sigma}. We repeat this process $2^{12}$ times and finally $2^{22}$ random vectors of the training set. The CRealNVP consists of six coupling layers and four hidden layers with 128 hidden nodes at each coupling layer for both  $\textbf{s}^{(j)}$ and  $\textbf{t}$. The activation functions of the hidden layers of the neural network are LeakyReLU functions. The neural networks are trained by minimizing the negative log-likelihood function with the ADAM optimizer, which ensures that the transformation $z=f_\Theta (y)$ follows the standard Gaussian distribution unconditionally on $\Theta$. 
After the training process, we obtain the PDF of the 2-dimensional gNTS distribution. 

\begin{table}[t]
\centering
\begin{tabular}{c|rrrr|cc}
\hline
 & \multicolumn{4}{c|}{gStdNTS parameters} & K-S & $p$-value\\
\hline
 Example 1 & $\alpha_1=1.25$ & $\theta_1=3.0$ & $\beta_1=0.0$ & $\rho = 0.0$ & $0.03668$ & $0.2712$ \\
           & $\alpha_2=1.25$ & $\theta_2=3.0$ & $\beta_2=0.0$ &   & & \\
\hline
 Example 2 & $\alpha_1=1.25$ & $\theta_1=3.0$ & $\beta_1=2.64$ & $\rho = -0.7$ & $0.03778$ & $0.2305$ \\
           & $\alpha_2=1.75$ & $\theta_2=5.0$ & $\beta_2=-4.49$ &   & & \\
\hline
 Example 3 & $\alpha_1=0.75$ & $\theta_1=1.0$ & $\beta_1=1.24$ & $\rho = 0.5$ & $0.03987$ & $0.1665$ \\
           & $\alpha_2=1.25$ & $\theta_2=3.0$ & $\beta_2=-2.64$ &   & & \\  
\hline
\end{tabular}
\caption{\label{Table:Example}Three examples of \gStdNTS~parameter sets}
\end{table}

As examples, we consider three sets of \gStdNTS~parameters excluded in the training set. Those parameters are presented in Table \ref{Table:Example}.
We simulate 1,000 random vectors for each parameter set in the table, respectively, and compare the sample with the distribution provided by CRealNVP trained for \gStdNTS~distribution. The 2-dimensional relative histogram of the simulated sample and the contour plot of PDFs generated by the CRealNVP method are exhibited in Figure \ref{fig:ExamplePDF} for the three-parameter sets, respectively. Graphically, the histogram and the contour plot have similar shapes. 
For the validation test of these three examples, we perform the Kolmogorov–Smirnov (K-S) test between the empirical CDF of the simulated samples and the CDF of the \gStdNTS~calculated by the trained CRealNVP method.
K-S statistic values for those three examples are presented in Table \ref{Table:Example} with $p$-values. Those three cases pass the K-S test, and there is no evidence that the empirical distribution is different from the \gStdNTS~distribution obtained by the CRealNVP method at the $5\%$ significant level in this investigation.
 
\begin{figure}
\centering
Example 1: 
$\alpha=(1.25, 1.25)^\tr$, $\theta=(3.0,3.0)^\tr$, $\beta = (0.0, 0.0)^\tr$, $\rho = 0.0$\\
\includegraphics[width = 8cm, height=6cm]{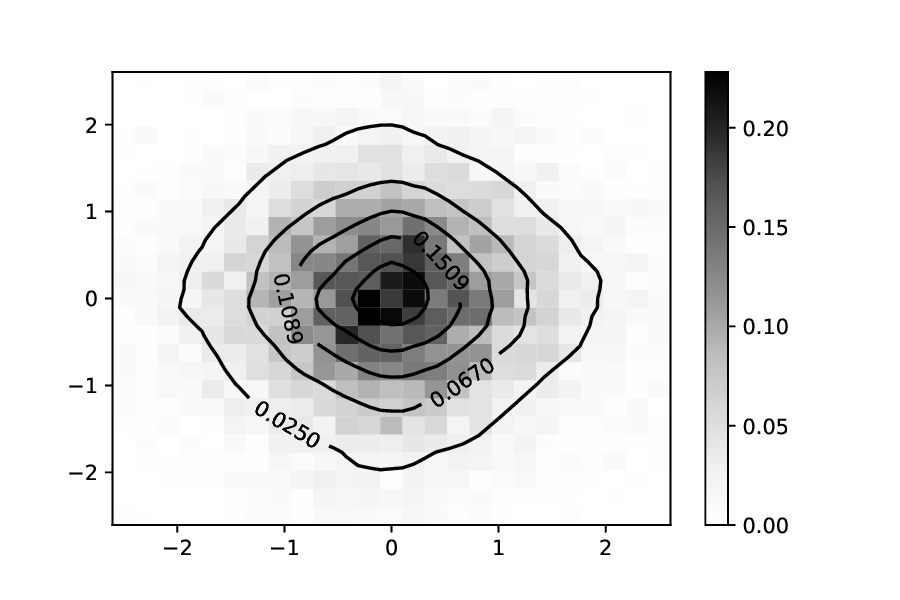}
\hspace{-1cm}
\includegraphics[width = 8cm, height=6cm]{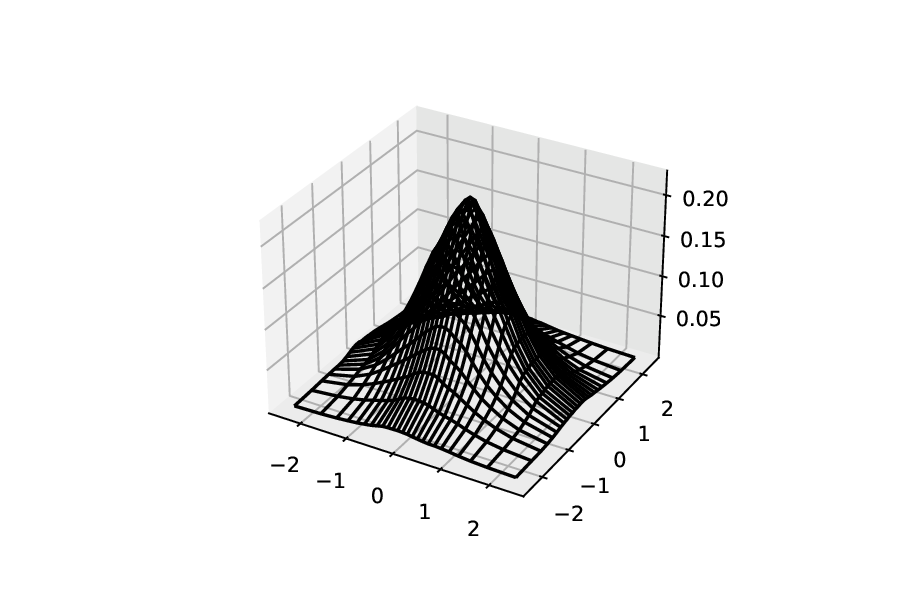}\\
~\\
Example 2:
$\alpha=(1.25, 1.75)^\tr$, $\theta=(3.0,5.0)^\tr$, $\beta = (2.64, -4.49)^\tr$, $\rho = -0.70$\\
\includegraphics[width = 8cm, height=6cm]{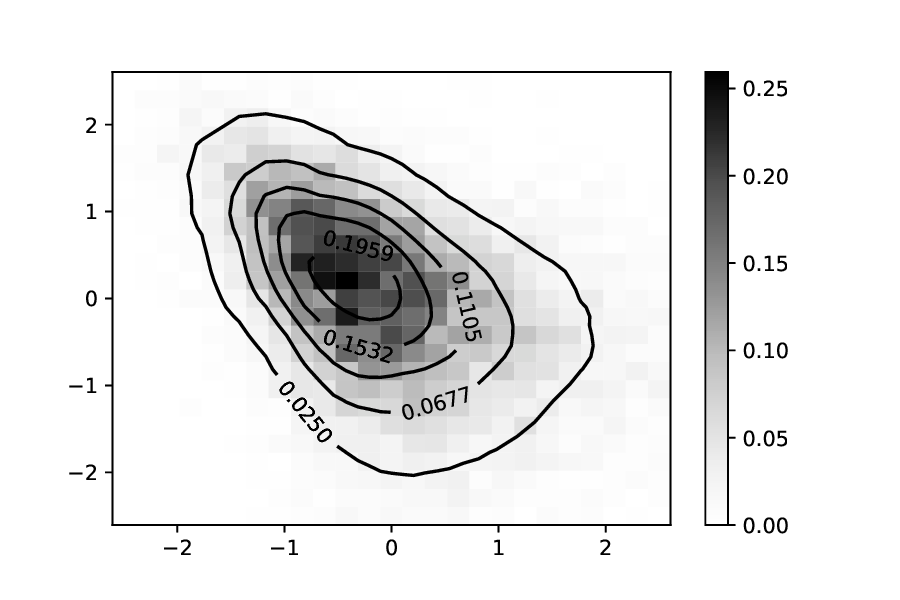}
\hspace{-1cm}
\includegraphics[width = 8cm, height=6cm]{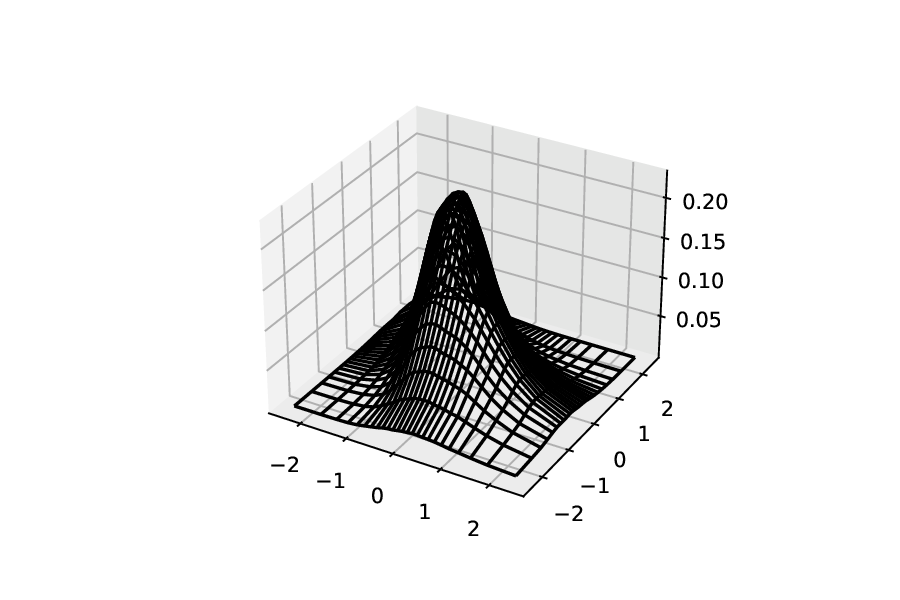}\\
~\\
Example 3:
$\alpha=(0.75, 1.25)^\tr$, $\theta=(1.0,3.0)^\tr$, $\beta = (1.24, -2.74)^\tr$, $\rho = 0.5$\\
\includegraphics[width = 8cm, height=6cm]{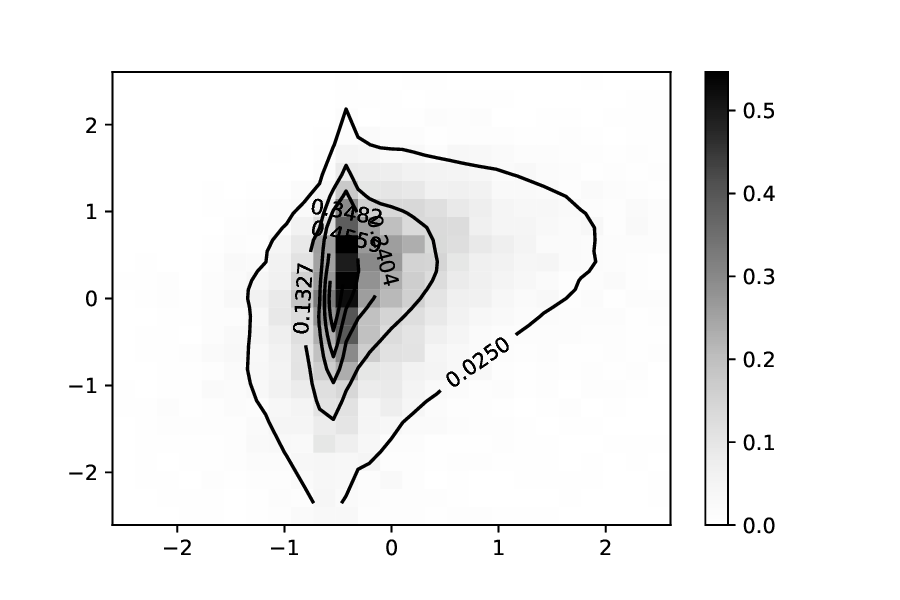}
\hspace{-1cm}
\includegraphics[width = 8cm, height=6cm]{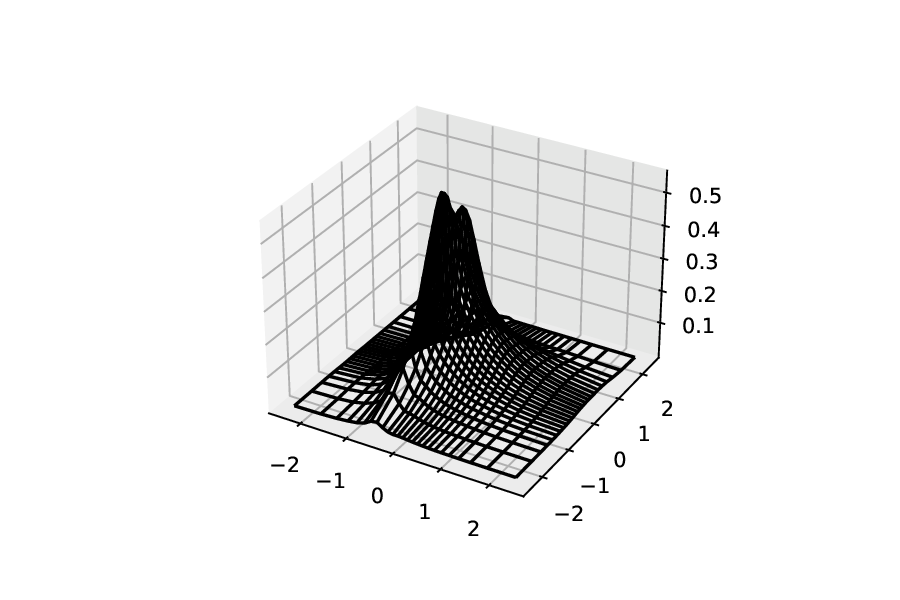}
\caption{\label{fig:ExamplePDF}Contour graphs (left) and 3d graphs (right) of the PDFs for the three examples of gStdNTS distribution.}
\end{figure}

\subsection{Parameter Estimation}
For an empirical illustration, we consider daily prices of the Japanese Yen (JPY)-U.S. Dollar (USD) exchange rate and USD-valued Nikkei225\footnote{Nihon Keizai Shinbun 225 index} (N225) from January 2, 2020 to December 29, 2023.  The USD-valued N225 prices are obtained by converting the original JPY-valued Nikkei225 levels into U.S. dollars using the JPY-USD exchange rate. Suppose $(F(t)))_{t\ge0}$ is the process of the JPY-USD exchange rate, such that one JPY to $F(t)$ dollar at time $t$, and $(V(t))_{t\ge0}$ is the dollar-valued price process of the N225, that is $V(t)=S(t)F(t)$, where $S(t)$ is the N225 at time $t$. We estimate market parameters for daily log-returns $X_F(t)$ and $X_V(t)$ of $V(t)$ and $F(t)$, respectively, as \eqref{eq:F&V in P}  in Section \ref{sec:QuantoOptionPricingOngNTS}.
As \eqref{eq:Y in P}, we set
with 
\[
(X(t))_{t\ge0}\sim
\gNTS_2\left(
\left(\begin{matrix} \alpha_F\\ \alpha_V\end{matrix}\right),
\left(\begin{matrix} \theta_F\\ \theta_V\end{matrix}\right),
\left(\begin{matrix} \beta_F\\ \beta_V\end{matrix}\right),
\left(\begin{matrix} \mu_F\\ \mu_V\end{matrix}\right),
\left(\begin{matrix} \sigma_F\\ \sigma_V\end{matrix}\right),
\left(\begin{matrix} 1 &  \rho \\ \rho & 1\end{matrix}\right)
\right)
\]
with $X(t) = (X_F(t), X_V(t))^\tr$.
We apply Proposition \ref{pro:gNTS normalization in time} for the daily time step $\varDelta t$, we have
\begin{equation}\label{eq:XtoZ}
\left(\begin{matrix}
X_F(\varDelta t)\\
X_V(\varDelta t)
\end{matrix}\right)
 \disteq  
\left(\begin{matrix}
m_F\\
m_V
\end{matrix}\right)
 + 
\left(\begin{matrix}
s_F & 0\\
0 & s_V
\end{matrix}\right)
\left(\begin{matrix}
\Xi_F(1)\\
\Xi_V(1)
\end{matrix}\right)
\end{equation}
where $\Xi(1)=(\Xi_F(1), \Xi_V(1))^\tr$ and
\[
\Xi 
\sim \gStdNTS_2\left(
\left(\begin{matrix} \alpha_F\\ \alpha_V\end{matrix}\right),
\left(\begin{matrix} \theta_{\Xi,F}\\ \theta_{\Xi,V} \end{matrix}\right),
\left(\begin{matrix} \beta_{\Xi,F} \\ \beta_{\Xi,V} \end{matrix}\right),
\left(\begin{matrix} 1 &  \rho \\ \rho & 1\end{matrix}\right)
\right)
\]
with
\begin{align*}
\theta_{\Xi,n}=\theta_n \varDelta t^{\frac{2}{\alpha_n}},~~~\beta_{\Xi,n}=\frac{\beta_n\varDelta t^{\frac{2}{\alpha_n}}}{s_n},~~~
m_n =\varDelta t\left(\mu_n +\frac{\alpha_n\beta_n}{2}\theta_n^{\frac{\alpha_n}{2}-1}\right),
\end{align*}
and
\begin{align*}
s_n = \sqrt{\varDelta t\frac{\alpha_n}{2}\theta_n^{\frac{\alpha_n}{2}-1}\left(\left(\frac{2-\alpha_n}{2\theta_n}\right)\beta_n^2+\sigma_n^2\right)},
\end{align*}
for $n\in\{F, V\}$.
We estimate $(m_F, m_V)^\tr$ and $(s_F, s_V)^\tr$ by the sample mean and sample standard deviation of the FX return and the index return, respectively.  Then, we fit the gStdNTS parameters of $\Xi$ using maximum likelihood estimation with the PDF trained by CRealNVP transformations.  

We repeat this parameter fit process for the other pair of a FX rate \& a market index such as British pound (GBP)-USD \& Financial Times Stock Exchange 100 Index (FTSE), Euro currency (EUR)-USD \& German stock Index\footnote{Deutscher Aktienindex} (DAX),
and Korean won (KRW)-USD \& KOSPI200 Index\footnote{Korean Composite Stock Price 200 Index} (KS200). The estimation results are also presented in Table \ref{Table:ParamFit}. The 2-dimensional histograms of the standardized log-returns of those 4 pairs are exhibited in Figure \ref{fig:MLEgNTS} together with the PDF contour map of gStdNTS distribution.
The estimated parameters for those four pairs of the FX rate and index returns are provided in the first row of Table \ref{Table:ParamFit}. 
In the last column of the table, we present the Kolmogorov–Smirnov (K-S) statistic and $p$-values for the goodness of fit test for the 2 dimensional CDF of empirical data and \gNTS~model,  respectively\footnote{Details of the K-S test for the 2-dimensional distribution in \cite{NAAMAN2021109088}.}.  In addition, the parameters of the NTS model are estimated as a benchmark model for the same data. In the NTS model, we assume that $X$ follows the NTS process and $X(\varDelta t) = m+\diag(s) \Xi(1)$ where $m=(m_F, m_V)^\tr$ and $s=(s_F, s_V)^\tr$ are the mean and standard deviation vectors, respectively, and $\Xi(1)$ is the standard NTS distributed with parameters $\left(\alpha, \theta, \beta_\Xi, \left(\begin{smallmatrix} 1&\rho\\ \rho & 1\end{smallmatrix}\right)\right)$ with $\alpha\in(0,2)$, $\theta>0$, $\beta_\Xi = (\beta_{\Xi,F}, \beta_{\Xi,V})^\tr$ and $\rho \in[-1,1]$.  More details on parameter fitting are described in the literature, including \cite{Kim:2022}. Comparing the K-S statistic, we see that the K-S statistic values for the NTS model are significantly larger than those of the \gNTS~model. According to the $p$-values, the estimated NTS distribution is rejected at the 5\% significance level, while the \gNTS~model is not rejected. That is, the performance of the parameters fitted to the \gNTS~model is much better than that of the NTS model.
\begin{sidewaystable}
\centering
\begin{footnotesize}
\begin{tabular}{c|c|c|ccc|c}
\hline
  & mean & Standard Deviation & \multicolumn{3}{c|}{Model Parameters} & K-S($p$-value)\\
\hline
  & & & \multicolumn{3}{c|}{\gStdNTS~~paramters} & \\
JPY-USD & $m_F = -2.766 \cdot 10^{-4}$ & $s_F = 5.812 \cdot 10^{-3}$ & $\alpha_{F} = 1.0171$ & $\theta_{\Xi,F} = 1.365$ & $\beta_{\Xi,F} = 1.978 \cdot 10^{-1}$ & $0.0409$ \\ 
N225 & $m_V = 9.847 \cdot 10^{-5}$ & $s_V = 1.357 \cdot 10^{-2}$ & $\alpha_V = 1.2300$ & $\theta_{\Xi,V} = 6.147 \cdot 10^{-1}$ & $\beta_{\Xi,V} = -2.188 \cdot 10^{-1}$ & $(0.1533)$ \\
 & & & \multicolumn{3}{c|}{$\rho = 0.3134$} & \\
\cline{4-6} \cline{6-7}
  & & & \multicolumn{3}{c|}{\stdNTS~~paramters} & \\
   &   &  & $\alpha = 0.7644$ & $\theta = 1.0719$ & $\beta_{\Xi,F} = 1.735 \cdot 10^{-1}$ & $0.3071$ \\ 
   & & & & & $\beta_{\Xi,V} = -7.530 \cdot 10^{-2}$ & $(0.0000)$ \\
 & & & \multicolumn{3}{c|}{$\rho = 0.3319$} & \\
\hline
  & & & \multicolumn{3}{c|}{\gStdNTS~~paramters} & \\
GBP-USD & $m_F = -2.913 \cdot 10^{-5}$ & $s_F = 6.211 \cdot 10^{-3}$ & $\alpha_{F} = 1.2373$ & $\theta_{\Xi,F} = 1.332$ & $\beta_{\Xi,F} = -1.940 \cdot 10^{-1}$ & $0.0288$ \\ 
FTSE & $m_V = -4.320 \cdot 10^{-6}$ & $s_V = 1.323 \cdot 10^{-2}$ & $\alpha_V = 0.9258$ & $\theta_{\Xi,V} = 2.892 \cdot 10^{-2}$ & $\beta_{\Xi,V} = -1.534 \cdot 10^{-2}$ & $(0.7530)$ \\
 & & & \multicolumn{3}{c|}{$\rho = 0.5366$} & \\
\cline{4-6} \cline{6-7}
  & & & \multicolumn{3}{c|}{\stdNTS~~paramters} & \\
   &   &  & $\alpha = 0.7195$ & $\theta = 0.9571$ & $\beta_{\Xi,F} = 5.590 \cdot 10^{-2}$ & $0.3143$ \\ 
   & & & & & $\beta_{\Xi,V} = -1.495 \cdot 10^{-1}$ & $(0.0000)$ \\
 & & & \multicolumn{3}{c|}{$\rho = 0.4379$} & \\
\hline
  & & & \multicolumn{3}{c|}{\gStdNTS~~paramters} & \\
EUR-USD & $m_F = -1.343 \cdot 10^{-5}$ & $s_F = 4.911 \cdot 10^{-3}$ & $\alpha_{F} = 1.3134$ & $\theta_{\Xi,F} = 4.341$ & $\beta_{\Xi,F} = -1.414 \cdot 10^{-1}$ & $0.0294$ \\ 
DAX & $m_V = 2.065 \cdot 10^{-4}$ & $s_V = 1.487 \cdot 10^{-2}$ & $\alpha_V = 0.9193$ & $\theta_{\Xi,V} = 1.783 \cdot 10^{-2}$ & $\beta_{\Xi,V} = -3.488 \cdot 10^{-3}$ & $(0.6852)$ \\
 & & & \multicolumn{3}{c|}{$\rho = 0.4192$} & \\
\cline{4-6} \cline{6-7}
  & & & \multicolumn{3}{c|}{\stdNTS~~paramters} & \\
   &   &  & $\alpha = 1.0756$ & $\theta = 1.0905$ & $\beta_{\Xi,F} = 6.332 \cdot 10^{-2}$ & $0.2947$ \\ 
   & & & & & $\beta_{\Xi,V} = -2.359 \cdot 10^{-1}$ & $(0.0000)$ \\
 & & & \multicolumn{3}{c|}{$\rho = 0.3225$} & \\
\hline
  & & & \multicolumn{3}{c|}{\gStdNTS~~paramters} & \\
KRW-USD & $m_F = -1.170 \cdot 10^{-4}$ & $s_F = 5.929 \cdot 10^{-3}$ & $\alpha_{F} = 1.1830$ & $\theta_{\Xi,F} = 1.090 \cdot 10^{1}$ & $\beta_{\Xi,F} = 8.177 \cdot 10^{-1}$ & $0.0315$ \\ 
KS200 & $m_V = 9.583 \cdot 10^{-5}$ & $s_V = 1.567 \cdot 10^{-2}$ & $\alpha_V = 1.3038$ & $\theta_{\Xi,V} = 4.696 \cdot 10^{-3}$ & $\beta_{\Xi,V} = -7.793 \cdot 10^{-3}$ & $(0.5658)$ \\
 & & & \multicolumn{3}{c|}{$\rho = 0.6399$} & \\
\cline{4-6} \cline{6-7}
  & & & \multicolumn{3}{c|}{\stdNTS~~paramters} & \\
   &   &  & $\alpha = 0.7631$ & $\theta = 1.8986$ & $\beta_{\Xi,F} = 3.070 \cdot 10^{-1}$ & $0.3454$ \\ 
   & & & & & $\beta_{\Xi,V} = -3.070 \cdot 10^{-1}$ & $(0.0000)$ \\
 & & & \multicolumn{3}{c|}{$\rho = 0.6185$} & \\
\hline
\end{tabular}
\end{footnotesize}
\caption{\label{Table:ParamFit}Results of parameter estimation of \gNTS~model to the 4 pairs of FX returns and foreign index returns, respectively}
\end{sidewaystable}

\begin{figure}
\centering
\includegraphics[width = 7cm]{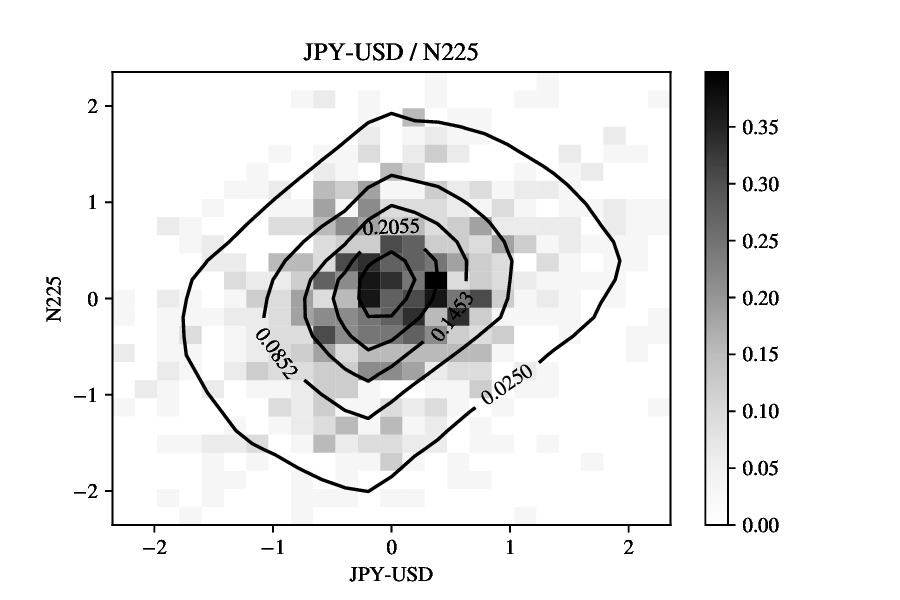}
\includegraphics[width = 7cm]{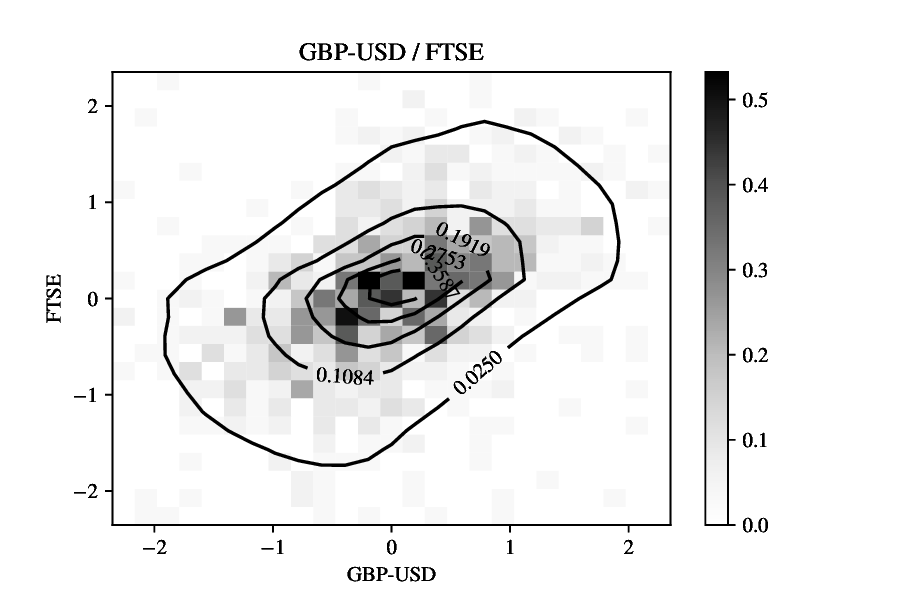}
\\
\includegraphics[width = 7cm]{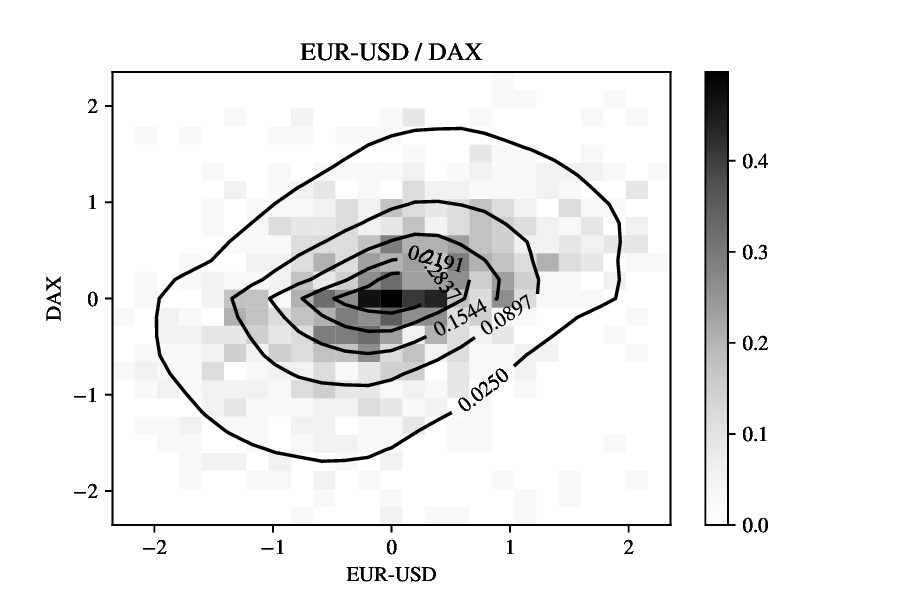}
\includegraphics[width = 7cm]{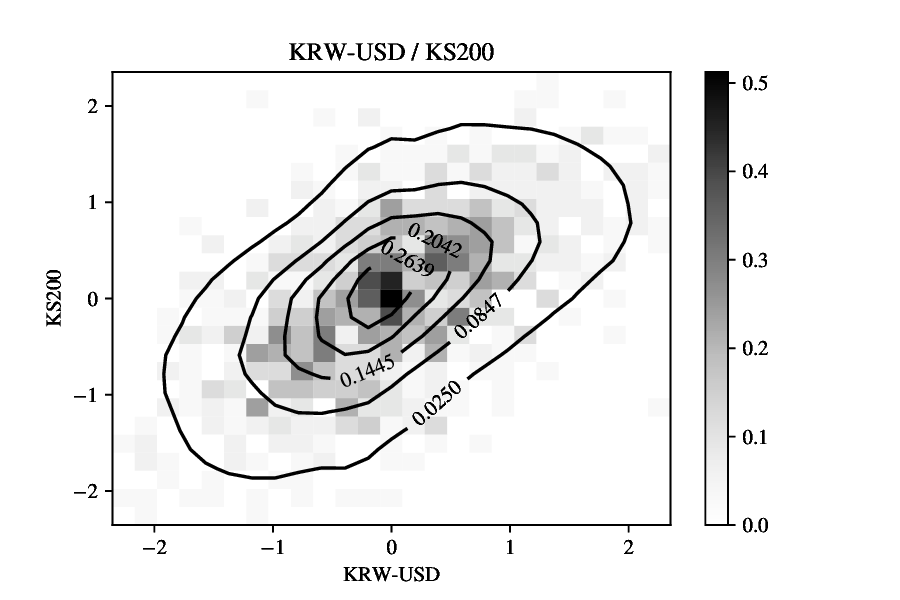}
\caption{\label{fig:MLEgNTS}Histograms and contour graphs of the PDFs of the standardized log returns.  The top-left is for JPY-USD and Nikkei 225 returns, the top-right is for GBP-USD and FTSE returns, the bottom-left is for EUR-USD and DAX returns, and the bottom-right is for KRW-USD and KS200 returns.}
\end{figure}

\subsection{Quanto Option Pricing}
Suppose \gNTS~parameters of $X$ is given by \eqref{eq:XtoZ}. By proposition \ref{prop:stdgnts to gnts} (a), we have
\[
(X(t))_{t\ge0}\sim
\gNTS_2\left(
\left(\begin{matrix} \alpha_F\\ \alpha_V\end{matrix}\right),
\left(\begin{matrix} \theta_F\\ \theta_V\end{matrix}\right),
\left(\begin{matrix} \beta_F\\ \beta_V\end{matrix}\right),
\left(\begin{matrix} \mu_F\\ \mu_V\end{matrix}\right),
\left(\begin{matrix} \sigma_F\\ \sigma_V\end{matrix}\right),
\left(\begin{matrix} 1 &  \rho \\ \rho & 1\end{matrix}\right)
\right)
\]
where
\begin{align*}
\theta_n = \frac{\theta_{\Xi,n}}{\varDelta t^{\frac{2}{\alpha_n}}}, ~ \beta_n =\frac{ \beta_{\Xi,n}s_n}{\varDelta t^{\frac{2}{\alpha_n}}},~\mu_n=\frac{m_n}{\varDelta t}-\frac{\alpha_n\beta_n}{2}\theta_n^{\frac{\alpha_n}{2}-1},
\sigma_n = \sqrt{\frac{2s_n^2\theta_n^{1-\frac{\alpha_n}{2}}}{\alpha_n \varDelta t}-\left(\frac{2-\alpha_n}{2\theta_n}\right)\beta_n^2}
\end{align*}
for $n\in\{F, V\}$.

Let $r = (r_F, r_V)^\tr$ with $r_F=r_d-r_f$ and $r_V = r_d$ to simplify notations.
There are infinitely many risk-neutral parameters $\hat\theta=(\hat\theta_F, \hat\theta_V)^\tr$ and $\hat\beta=(\hat\beta_F, \hat\beta_V)^\tr$ satisfying \textbf{RN.2},
which is equivalent to 
\begin{align*}
\hat\beta_n=\hat\theta_n-\frac{\sigma_n^2}{2}-\left(\mu_n-r_n+\hat\theta_n^{\frac{\alpha_n}{2}}\right)^{\frac{2}{\alpha_n}}, ~~~\text{ for } n\in\{F, V\}.
\end{align*}
To select one set of risk-neutral parameters, we try to find the parameter set which is as close to the physical parameters as possible.
That is we find $\hat\theta^*=(\hat\theta^*_F, \hat\theta^*_V)^\tr$ and $\hat\beta^*=(\hat\beta^*_F, \hat\beta^*_V)^\tr$ close to the physical parameters $\theta$ and $\beta$ as follows:
\[
(\hat\theta^*_n, \hat\beta^*_n) = \argmin_{(\hat\theta_n, \hat\beta_n)} \sqrt{\left(\hat\theta_n-\theta_n\right)^2+\left(\hat\beta_n-\beta_n\right)^2}, ~~~\text{ for } n\in\{F, V\}.
\]
Then we obtain
\[
(X(t))_{t\ge0}\sim
\gNTS_2\left(
\left(\begin{matrix} \alpha_F\\ \alpha_V\end{matrix}\right),
\left(\begin{matrix} \hat\theta^*_F\\ \hat\theta^*_V\end{matrix}\right),
\left(\begin{matrix} \hat\beta^*_F\\ \hat\beta^*_V\end{matrix}\right),
\left(\begin{matrix} \mu_F\\ \mu_V\end{matrix}\right),
\left(\begin{matrix} \sigma_F\\ \sigma_V\end{matrix}\right),
\left(\begin{matrix} 1 &  \rho \\ \rho & 1\end{matrix}\right)
\right)
\]
under the risk-neutral measure $\Q_{\hat\theta^*, \hat\beta^*}$.
Recall the quanto call option pricing formula \eqref{eq:QuantoOptionPrice Expectation form} in Section \ref{sec:QuantoOptionPricingOngNTS},
we consider the quanto call option with the time to maturity $T$, strike price $K$ and the fixed exchange rate $F_\text{fix}$. 
To calculate the quanto call price, we must know the distribution of $(X_F(T),X_V(T))^\tr$ for the time to maturity $T$.
We apply Proposition \ref{pro:gNTS normalization in time} to $X$ under the risk-neutral measure $\Q_{\hat\theta^*, \hat\beta^*}$, we obtain as follows:
\[
\left(\begin{matrix}
X_F(T)\\
X_V(T)
\end{matrix}\right)
 \disteq  
\left(\begin{matrix}
\hat m_F\\
\hat m_V
\end{matrix}\right)
 + 
\left(\begin{matrix}
\hat s_F & 0\\
0 & \hat s_V
\end{matrix}\right)
\left(\begin{matrix}
\hat \Xi_F(1)\\
\hat \Xi_V(1)
\end{matrix}\right)
\]
with
\[
\hat \Xi =(\hat \Xi_F,\hat \Xi_V)^\tr
\sim \gStdNTS_2\left(
\left(\begin{matrix} \alpha_F\\ \alpha_V\end{matrix}\right),
\left(\begin{matrix} \hat\theta_{\hat\Xi,F}\\ \hat\theta_{\hat\Xi,V} \end{matrix}\right),
\left(\begin{matrix} \hat\beta_{\hat\Xi,F} \\ \hat\beta_{\hat\Xi,V} \end{matrix}\right),
\left(\begin{matrix} 1 &  \rho \\ \rho & 1\end{matrix}\right)
\right)
\]
where
\begin{align*}
\hat\theta_{\hat\Xi,n}=\hat\theta^*_n T^{\frac{2}{\alpha_n}},
~~~
\hat\beta_{\hat\Xi,n}=\frac{\hat\beta^*_n T^{\frac{2}{\alpha_n}}}{s_n},
~~~
\hat s_n = \sqrt{\frac{\alpha_nT}{2}\left(\hat\theta^*_n\right)^{\frac{\alpha_n}{2}-1}\left(\left(\frac{2-\alpha_n}{2\hat\theta^*_n}\right)\left(\hat\beta^*_n\right)^2+\sigma_n^2\right)},
\end{align*}
and
\begin{align*}
&\hat m_n =T\left(r_n+\left(\hat\theta^*_n-\hat\beta^*_n-\frac{\sigma_n^2}{2}\right)^{\frac{\alpha_n}{2}}+\left(\frac{\alpha_n\hat\beta^*_n}{2\hat\theta^*_n}-1\right)\left(\hat\theta^*_n\right)^{\frac{\alpha_n}{2}}\right), ~~~\text{ for }~~~ n\in\{F, V\}.
\end{align*}
Using the parameters, we continue option pricing of the equation \eqref{eq:QuantoOptionPrice Expectation form}:
\begin{align*}
&E_{\Q_{\hat\theta^*, \hat\beta^*}}\left[e^{-r_d T}F_\text{fix}(S(T) - K)^+\right]
\\
&
=e^{-r_d T}F_\text{fix}E_{\Q_{\hat\theta^*, \hat\beta^*}}\left[(S(0)\exp(\hat m_F+\hat s_F\hat\Xi_F(1)-\hat m_V-\hat s_V\hat\Xi_V(1)) - K)^+\right].
\end{align*}
Let 
\[
H(\hat\Xi(1))=\left(e^{\hat s_F\hat\Xi_F(1)-\hat s_V\hat\Xi_V(1)} - \frac{K}{S(0)e^{\hat m_F-\hat m_V}}\right)^+.
\]
Then we have
\[
E_{\Q_{\hat\theta^*, \hat\beta^*}}\left[e^{-r_d T}F_\text{fix}(S(T) - K)^+\right]
= e^{-r_d T}F_\text{fix}S(0)e^{\hat m_F-\hat m_V}E_{\Q_{\hat\theta^*, \hat\beta^*}}\left[H(\hat\Xi(1))\right].
\]
By the CRealNVP model, we set $\hat\Xi(1) = f_\Theta^{-1}(Z)$ with the standard normal $Z$ and the parameters $\Theta =\left(\alpha_F, \alpha_V, \hat\theta_{\hat\Xi,F}, \hat\theta_{\hat\Xi,V}, \hat\beta_{\hat\Xi,F}, \hat\beta_{\hat\Xi,V},\rho \right)$.
Therefore, we have
\[
E_{\Q_{\hat\theta^*, \hat\beta^*}}\left[H(\hat\Xi(1))\right]=E_{\Q_{\hat\theta^*, \hat\beta^*}}\left[H\left( f_\Theta^{-1}(Z) \right)\right]
=\iint H\left(f_\Theta^{-1}(z)\right) p_Z(z)dz,
\]
where $p_Z(z)$ is the PDF of the 2-dimensional standard normal distribution
 The inverse function $f_\Theta^{-1}$ is obtained by the definition of the CRealNVP model as $f_\Theta^{-1}=(f_\Theta^{(1)})^{-1}\circ(f_\Theta^{(2)})^{-1}\circ\cdots\circ(f_\Theta^{(J)})^{-1}$ where
 \[
 \left(f_\Theta^{(j)}\right)^{-1}(z)=b^{(j)}\odot z+\left(I-b^{(j)}\right)\odot\left(z\odot \exp\left(\textbf{s}^{(j)}(b^{(j)}\odot z)\right)+\textbf{t}^{(j)}\left(b^{(j)}\odot z\right)  \right).
 \]
The double integral can be approximated by a numerical integration.

\begin{table}[t]
\centering
\begin{tabular}{c|cccc}
 \hline
 & $S(0)$ & $F_\text{fix}$ & $r_f$ \\
 \hline
JPY-USD/N225 & $33464.2$ & $7.071\cdot 10^{-3}$ & $-0.1\%$\\
GBP-USD/FTSE & $7733.20$ & $1.273$ & $5.25\%$\\
EUR-USD/DAX & $16751.6$ & $1.107$ & $4.5\%$\\
KRW-USD/KS200 & $357.990$ & $7.826\cdot 10^{-4}$ &$3.5\%$ \\
\hline
\end{tabular}
\caption{\label{Table:FixedValues}Market information on December 29, 2023.  $S(0)$,  $F_\text{fix}$, and $r_f$ mean the foreign index, FX rate,  and foreign risk-free rate with respect to the 4 pairs of the examples, respectively}
\end{table}

For example, we calculate $(\hat m_F, \hat m_V)^\tr$, $(\hat s_F, \hat s_V)^\tr$, $(\hat \theta_{\hat\Xi,F}, \hat \theta_{\hat\Xi,V})^\tr$ and $(\hat \beta_{\hat\Xi,F}, \hat \beta_{\hat\Xi,V})^\tr$ for $T\in\{1/52 (1$-week$)$, $2/52 (2$-weeks$)$, $3/52 (3$-weeks$)$, $4/52 (4$-weeks$)\}$ based on the estimated parameters in Table \ref{Table:ParamFit} for the 4 cases (JPY-USD/N225, GBP-USD/FTSE, EUR-USD/DAX, KRW-USD/KS200). In this calculation, we set $r_d=5.5\%$ which is the U.S standard rate (Federal Fund Rate), and we set $r_f$ in $\{-0.1\%$,  $5.25\%$, $4.5\%$, $3.5\%\}$ which are standard rates of Japan, U.K., European Union, and Korea, respectively, on December 2023 (See Table \ref{Table:FixedValues}). The risk neutral parameters based on this calculation are presented in Table \ref{Table:RNParameterSet}. 
Moreover, the values of $S(0)$ and $F_\text{fix}$ are presented in Table \ref{Table:FixedValues} which were observed on December 29, 2023.

The quanto call option prices for the 4 cases with the risk-neutral parameters in Table \ref{Table:RNParameterSet} for time to maturities 1-4 weeks are calculated and presented in Figure \ref{Figure:QuantoCallPrices}. Since the index prices are all different, we use the moneyness $M = K/S(0)$ instead of the strike price $K$, and change the function $H$ to
\[
H(\hat\Xi(1))=\left(e^{\hat s_F\hat\Xi_F(1)-\hat s_V\hat\Xi_V(1)} - M e^{-\hat m_F+\hat m_V}\right)^+.
\]
For this reason, the $x$-axes of the 4 plates of Figure \ref{Figure:QuantoCallPrices} present the moneyness $K/S(0)$.
\begin{sidewaystable}
\centering
\begin{footnotesize}
\begin{tabular}{c|c|l|l|ll}
\hline
  & $T$ (week) & mean & Standard Deviation & \multicolumn{2}{c}{Risk-neutral \gStdNTS~parameters} \\
\hline
JPY-USD
 & 1 & $\hat m_F = 9.936 \cdot 10^{-4}$ & $\hat s_F = 1.291 \cdot 10^{-2}$ & $\hat\theta_{\hat\Xi,F} = 3.040 \cdot 10^{1}$ & $\hat\beta_{\hat\Xi,F} = 3.889$ \\ 
N225
 &  & $\hat m_V = -4.661 \cdot 10^{-4}$ & $\hat s_V = 2.993 \cdot 10^{-2}$ & $\hat \theta_{\hat\Xi,V} = 8.001$ & $\hat \beta_{\hat\Xi,V} = -1.405$ \\
\cline{2-6}
 & 2 & $\hat m_F = 1.987 \cdot 10^{-3}$ & $\hat s_F = 1.825 \cdot 10^{-2}$ & $\hat\theta_{\hat\Xi,F} = 1.188 \cdot 10^{2}$ & $\hat\beta_{\hat\Xi,F} = 1.074 \cdot 10^{1}$ \\ 
 &  & $\hat m_V = -9.322 \cdot 10^{-4}$ & $\hat s_V = 4.232 \cdot 10^{-2}$ & $\hat \theta_{\hat\Xi,V} = 2.470 \cdot 10^{1}$ & $\hat \beta_{\hat\Xi,V} = -3.067$ \\
\cline{2-6}
 & 3 & $\hat m_F = 2.981 \cdot 10^{-3}$ & $\hat s_F = 2.236 \cdot 10^{-2}$ & $\hat\theta_{\hat\Xi,F} = 2.637 \cdot 10^{2}$ & $\hat\beta_{\hat\Xi,F} = 1.947 \cdot 10^{1}$ \\ 
 &  & $\hat m_V = -1.398 \cdot 10^{-3}$ & $\hat s_V = 5.183 \cdot 10^{-2}$ & $\hat \theta_{\hat\Xi,V} = 4.775 \cdot 10^{1}$ & $\hat \beta_{\hat\Xi,V} = -4.841$ \\
\cline{2-6}
 & 4 & $\hat m_F = 3.974 \cdot 10^{-3}$ & $\hat s_F = 2.581 \cdot 10^{-2}$ & $\hat\theta_{\hat\Xi,F} = 4.643 \cdot 10^{2}$ & $\hat\beta_{\hat\Xi,F} = 2.969 \cdot 10^{1}$ \\ 
 &  & $\hat m_V = -1.864 \cdot 10^{-3}$ & $\hat s_V = 5.985 \cdot 10^{-2}$ & $\hat \theta_{\hat\Xi,V} = 7.623 \cdot 10^{1}$ & $\hat \beta_{\hat\Xi,V} = -6.694$ \\
\hline
GBP-USD
 & 1 & $\hat m_F = -4.535 \cdot 10^{-5}$ & $\hat s_F = 1.367 \cdot 10^{-2}$ & $\hat\theta_{\hat\Xi,F} = 1.707 \cdot 10^{1}$ & $\hat\beta_{\hat\Xi,F} = -1.097$ \\ 
FTSE
 &  & $\hat m_V = 5.886 \cdot 10^{-4}$ & $\hat s_V = 2.906 \cdot 10^{-2}$ & $\hat \theta_{\hat\Xi,V} = 8.747 \cdot 10^{-1}$ & $\hat \beta_{\hat\Xi,V} = -1.692 \cdot 10^{-1}$ \\
\cline{2-6}
 & 2 & $\hat m_F = -9.069 \cdot 10^{-5}$ & $\hat s_F = 1.933 \cdot 10^{-2}$ & $\hat\theta_{\hat\Xi,F} = 5.235 \cdot 10^{1}$ & $\hat\beta_{\hat\Xi,F} = -2.378$ \\ 
 &  & $\hat m_V = 1.177 \cdot 10^{-3}$ & $\hat s_V = 4.109 \cdot 10^{-2}$ & $\hat \theta_{\hat\Xi,V} = 3.910$ & $\hat \beta_{\hat\Xi,V} = -5.347 \cdot 10^{-1}$ \\
\cline{2-6}
 & 3 & $\hat m_F = -1.360 \cdot 10^{-4}$ & $\hat s_F = 2.368 \cdot 10^{-2}$ & $\hat\theta_{\hat\Xi,F} = 1.008 \cdot 10^{2}$ & $\hat\beta_{\hat\Xi,F} = -3.740$ \\ 
 &  & $\hat m_V = 1.766 \cdot 10^{-3}$ & $\hat s_V = 5.033 \cdot 10^{-2}$ & $\hat \theta_{\hat\Xi,V} = 9.389$ & $\hat \beta_{\hat\Xi,V} = -1.048$ \\
\cline{2-6}
 & 4 & $\hat m_F = -1.814 \cdot 10^{-4}$ & $\hat s_F = 2.734 \cdot 10^{-2}$ & $\hat\theta_{\hat\Xi,F} = 1.605 \cdot 10^{2}$ & $\hat\beta_{\hat\Xi,F} = -5.157$ \\ 
 &  & $\hat m_V = 2.354 \cdot 10^{-3}$ & $\hat s_V = 5.812 \cdot 10^{-2}$ & $\hat \theta_{\hat\Xi,V} = 1.748 \cdot 10^{1}$ & $\hat \beta_{\hat\Xi,V} = -1.690$ \\
\hline
EUR-USD
 & 1 & $\hat m_F = 1.339 \cdot 10^{-4}$ & $\hat s_F = 1.081 \cdot 10^{-2}$ & $\hat\theta_{\hat\Xi,F} = 4.800 \cdot 10^{1}$ & $\hat\beta_{\hat\Xi,F} = -6.043 \cdot 10^{-1}$ \\ 
DAX
 &  & $\hat m_V = 3.291 \cdot 10^{-4}$ & $\hat s_V = 3.279 \cdot 10^{-2}$ & $\hat \theta_{\hat\Xi,V} = 5.525 \cdot 10^{-1}$ & $\hat \beta_{\hat\Xi,V} = -8.136 \cdot 10^{-2}$ \\
\cline{2-6}
 & 2 & $\hat m_F = 2.678 \cdot 10^{-4}$ & $\hat s_F = 1.529 \cdot 10^{-2}$ & $\hat\theta_{\hat\Xi,F} = 1.379 \cdot 10^{2}$ & $\hat\beta_{\hat\Xi,F} = -1.228$ \\ 
 &  & $\hat m_V = 6.581 \cdot 10^{-4}$ & $\hat s_V = 4.637 \cdot 10^{-2}$ & $\hat \theta_{\hat\Xi,V} = 2.496$ & $\hat \beta_{\hat\Xi,V} = -2.599 \cdot 10^{-1}$ \\
\cline{2-6}
 & 3 & $\hat m_F = 4.017 \cdot 10^{-4}$ & $\hat s_F = 1.872 \cdot 10^{-2}$ & $\hat\theta_{\hat\Xi,F} = 2.557 \cdot 10^{2}$ & $\hat\beta_{\hat\Xi,F} = -1.859$ \\ 
 &  & $\hat m_V = 9.872 \cdot 10^{-4}$ & $\hat s_V = 5.679 \cdot 10^{-2}$ & $\hat \theta_{\hat\Xi,V} = 6.031$ & $\hat \beta_{\hat\Xi,V} = -5.127 \cdot 10^{-1}$ \\
\cline{2-6}
 & 4 & $\hat m_F = 5.356 \cdot 10^{-4}$ & $\hat s_F = 2.162 \cdot 10^{-2}$ & $\hat\theta_{\hat\Xi,F} = 3.963 \cdot 10^{2}$ & $\hat\beta_{\hat\Xi,F} = -2.495$ \\ 
 &  & $\hat m_V = 1.316 \cdot 10^{-3}$ & $\hat s_V = 6.557 \cdot 10^{-2}$ & $\hat \theta_{\hat\Xi,V} = 1.128 \cdot 10^{1}$ & $\hat \beta_{\hat\Xi,V} = -8.302 \cdot 10^{-1}$ \\
\hline
KRW-USD
 & 1 & $\hat m_F = 2.992 \cdot 10^{-4}$ & $\hat s_F = 1.307 \cdot 10^{-2}$ & $\hat\theta_{\hat\Xi,F} = 1.570 \cdot 10^{2}$ & $\hat\beta_{\hat\Xi,F} = 6.232$ \\ 
KS200
 &  & $\hat m_V = 8.016 \cdot 10^{-5}$ & $\hat s_V = 3.460 \cdot 10^{-2}$ & $\hat \theta_{\hat\Xi,V} = 5.286 \cdot 10^{-2}$ & $\hat \beta_{\hat\Xi,V} = -4.584 \cdot 10^{-2}$ \\
\cline{2-6}
 & 2 & $\hat m_F = 5.985 \cdot 10^{-4}$ & $\hat s_F = 1.848 \cdot 10^{-2}$ & $\hat\theta_{\hat\Xi,F} = 5.069 \cdot 10^{2}$ & $\hat\beta_{\hat\Xi,F} = 1.423 \cdot 10^{1}$ \\ 
 &  & $\hat m_V = 1.603 \cdot 10^{-4}$ & $\hat s_V = 4.893 \cdot 10^{-2}$ & $\hat \theta_{\hat\Xi,V} = 1.531 \cdot 10^{-1}$ & $\hat \beta_{\hat\Xi,V} = -9.388 \cdot 10^{-2}$ \\
\cline{2-6}
 & 3 & $\hat m_F = 8.977 \cdot 10^{-4}$ & $\hat s_F = 2.263 \cdot 10^{-2}$ & $\hat\theta_{\hat\Xi,F} = 1.006 \cdot 10^{3}$ & $\hat\beta_{\hat\Xi,F} = 2.305 \cdot 10^{1}$ \\ 
 &  & $\hat m_V = 2.405 \cdot 10^{-4}$ & $\hat s_V = 5.992 \cdot 10^{-2}$ & $\hat \theta_{\hat\Xi,V} = 2.851 \cdot 10^{-1}$ & $\hat \beta_{\hat\Xi,V} = -1.428 \cdot 10^{-1}$ \\
\cline{2-6}
 & 4 & $\hat m_F = 1.197 \cdot 10^{-3}$ & $\hat s_F = 2.613 \cdot 10^{-2}$ & $\hat\theta_{\hat\Xi,F} = 1.636 \cdot 10^{3}$ & $\hat\beta_{\hat\Xi,F} = 3.247 \cdot 10^{1}$ \\ 
 &  & $\hat m_V = 3.207 \cdot 10^{-4}$ & $\hat s_V = 6.919 \cdot 10^{-2}$ & $\hat \theta_{\hat\Xi,V} = 4.433 \cdot 10^{-1}$ & $\hat \beta_{\hat\Xi,V} = -1.922 \cdot 10^{-1}$ \\
\hline
\end{tabular}
\end{footnotesize}
\caption{\label{Table:RNParameterSet}Risk-neural parameters minimize the distance from the historically estimated physical parameters for the 4 pairs of examples, respectively.}
\end{sidewaystable}

\begin{figure}
\centering
\includegraphics[width = 7cm]{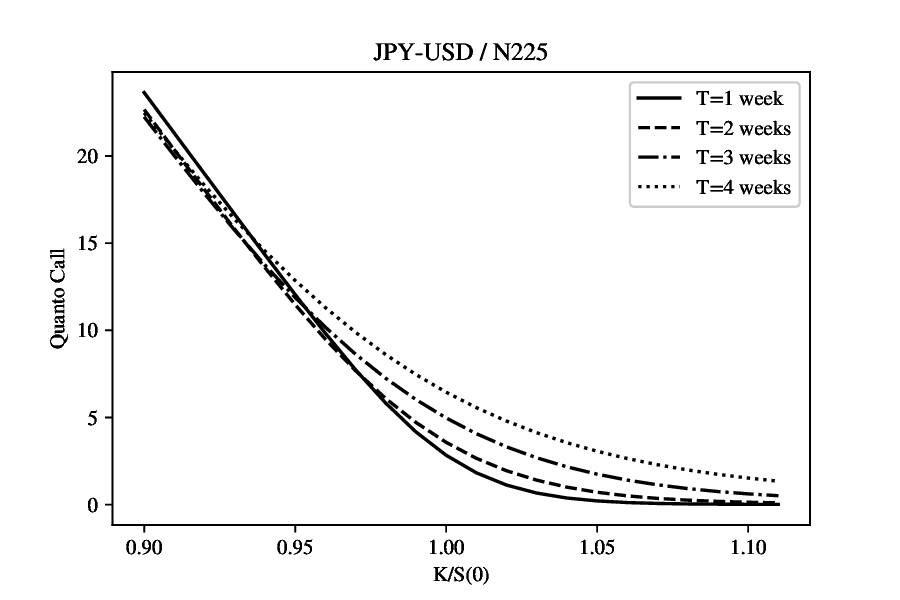}
\includegraphics[width = 7cm]{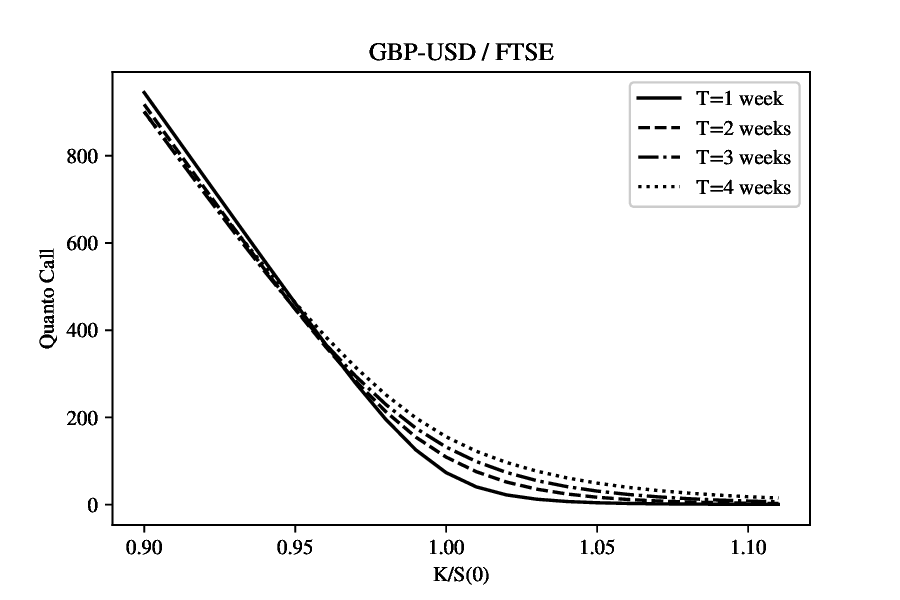}\\
\includegraphics[width = 7cm]{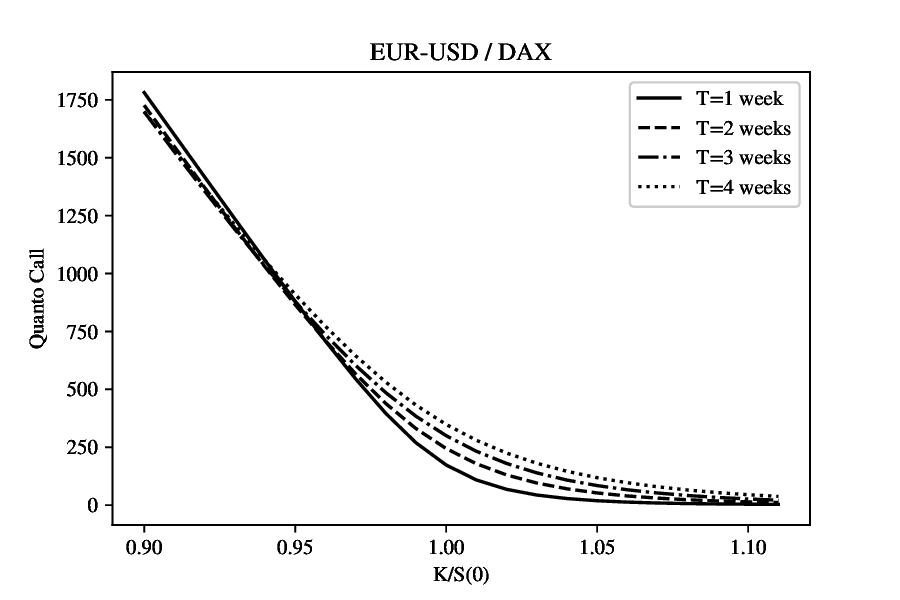}
\includegraphics[width = 7cm]{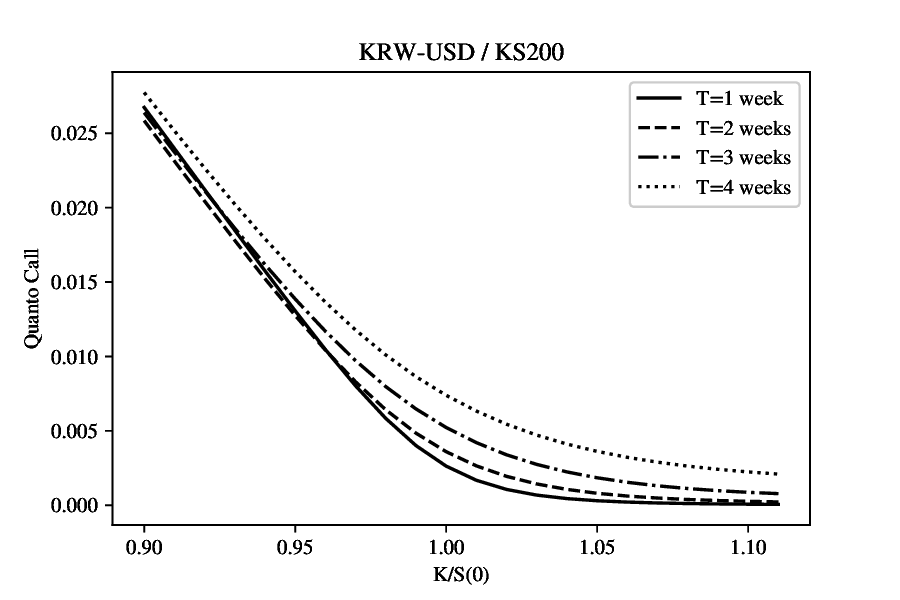}
\\
\caption{\label{Figure:QuantoCallPrices}Quanto Call Prices.  The top-left is for JPY-USD/N225, the top-right is for GBP-USD/FTSE, the bottom-left is for EUR-USD/DAX and the bottom-right is for KRW-USD/KS200 returns.}
\end{figure}

\section{Conclusion}
We have discussed a method for pricing European quanto options based on the \gNTS~model. 
The \gNTS~process captures both the fat-tail property and asymmetric dependence between returns of an FX rate and a corresponding foreign index.
Different from the NTS process, the \gNTS~process allows different subordinators to the foreign index and FX return distributions, respectively,
and it describes different volatility characteristics for the FX rate and foreign index that the NTS model does not capture. 

Since the \gNTS~ does not have a simple analytic form of distribution,
we use the CRealNVP model to find the PDF of the \gNTS~process. 
In this study, we construct the CRealNVP model for the \gStdNTS~process and train the model using the training set generated by the Monte-Carlo simulation of the \gStdNTS~process.
The \gNTS~process can be decomposed by the mean, standard deviation, and \gStdNTS~process.
We empirically fit the \gStdNTS~process parameters to the 4 pairs of FX rate and foreign index data: USD-JPY/N225, USD-GBP/FTSE, USD-EUR/DAX, and UDS-KRW/KS200.
According to the K-S test in this investigation, the parameter estimation for \gNTS~model performs better than that for the NTS model, which is the benchmark model.

Applying Sato's theorem and Girsanov's theorem to the time-changed Brownian motion model, a risk-neutral measure of the \gNTS~model is obtained.
Since there are infinitely many risk-neutral measures in \gNTS~model,
we select one risk-neutral measure whose parameter set has the smallest distance from the parameter set of the physical measure.
This method was applied to the 4 example pairs of the empirical data, and the risk-neutral parameters are obtained for each case.
Using the risk-neutral parameters, we successfully calculate prices of the example quanto option for the 4 pairs of FX rates and foreign market indexes.

We conclude that the distribution of \gNTS~process is successfully obtained by the CRealNVP model.
Using this method, we can fit the \gNTS~process to the empirical data efficiently.
Moreover, we can find the risk-neutral measure of the \gNTS~model and calculate the price of the quanto option using the CRealNVP model with the risk-neutral parameters.

\section{Appendix}
\begin{proof}[Proof of Proposition \ref{pro:gNTS normalization in time}] Let
\[
Y\sim\gNTS_N(\alpha, \theta_Y, \beta_Y, \mu_Y, \sigma_Y, R),
\]
where $n$-th elements of $\theta_Y$, $\beta_Y$, $\mu_Y$, and $\sigma_Y$ are 
$\theta_{Y,n}=\theta_n T^{\frac{2}{\alpha_n}}$, 
~~~
$\beta_{Y,n}=\beta_n T^{\frac{2}{\alpha_n}}$,
~~~
$\mu_{Y,n} = \mu_n T$, and $\sigma_{Y,n}= \sigma_n T^{\frac{1}{\alpha_n}}$.
Then we have $X(T) \disteq Y(1)$.
Moreover, by the Proposition \ref{prop:stdgnts to gnts}, we have $Y(1)\disteq m+\diag(s) \Xi(1)$
with $\Xi\sim \gStdNTS_N(\alpha,\theta_\Xi,\beta_\Xi,R)$, 
where $\theta_\Xi=\theta_Y$ and the $n$-th elements of $m\in\R^N$, $s\in\R_+^N$ and $\beta_\Xi\in\R^N$ are
\begin{align*}
m_n &=\mu_n T+\frac{\alpha_n\beta_n T^{\frac{2}{\alpha_n}}}{2}\left(\theta_n T^{\frac{2}{\alpha_n}}\right)^{\frac{\alpha_n}{2}-1}
=T\left(\mu_n+\frac{\alpha_n\beta_n}{2}\theta_n^{\frac{\alpha_n}{2}-1}\right)
,\\
%
s_n &= \sqrt{\frac{\alpha_n}{2}\left(\theta_nT^{\frac{2}{\alpha_n}}\right)^{\frac{\alpha_n}{2}-1}\left(\left(\frac{2-\alpha_n}{2\theta_nT^{\frac{2}{\alpha_n}}}\right)\left(\beta_nT^{\frac{2}{\alpha_n}}\right)^2+\left(\sigma_nT^{\frac{1}{\alpha_n}}\right)^2\right)}\\
&
= \sqrt{\frac{\alpha_n}{2}\theta_n^{\frac{\alpha_n}{2}-1}T\left(\left(\frac{2-\alpha_n}{2\theta_n}\right)\beta_n^2+\sigma_n^2\right)}
\end{align*}
and
\begin{align*}
\beta_{\Xi,n}=\beta_n T^{\frac{2}{\alpha}}/s_n,
\end{align*}
respectively. Here, $\alpha_n$, $\theta_n$, $\beta_n$, $\mu_n$, and $\sigma_n$ are the $n$-th elements of $\alpha$, $\theta$, $\beta$, $\mu$, and $\sigma$, respectively.
\end{proof}

\begin{proof}[Proof of Proposition \ref{pro:GirsanovGNTS}]
Let $\hat \beta=(\hat \beta_1,\hat \beta_2,\cdots,\hat \beta_N)^\tr \in \R^N$ and  $(H(t))_{t\ge0}$ be an $N$-dimensional process satisfying
\[
\textup{diag}(\sigma) R^{\frac{1}{2}} H(t)=\diag \left(\tau(t)^\sqt\right)(\beta-\hat \beta)
\]
with $H(t) = \left(H_1(t), H_2(t), \cdots, H_N(t)\right)^\tr$. 
Then we have
\begin{align*}
X(t) &= \mu t + \diag\left(\hat \beta\right) \int_0^t\tau(u)du \\
       & +  \,\, \textup{diag}(\sigma)  \left(\int_0^t \diag \left(\tau^\sqt(t)\right) R^{\frac{1}{2}} H(u)du+\int_0^t \diag \left(\tau^\sqt(t)\right)d(R^{\frac{1}{2}}B(u))\right)\\
      &= \mu t + \diag\left(\hat \beta\right) \int_0^t\tau(u)du  + \textup{diag}(\sigma)  \int_0^t \diag \left(\tau^\sqt(t)\right) R^{\frac{1}{2}} \left(H(u)du+dB(u)\right).
\end{align*}
With
\begin{equation}\label{eq:RNDer}
\frac{d\Q_{\hat\beta}}{d\P}=e^{\Xi(T)-\frac{1}{2}[\Xi,\Xi](T)},
~~~\text{ for }~~
\Xi(t)=-\sum_{n=1}^N\int_0^tH_n(s)dB_n(s),
\end{equation}
by Girsanov's theorem (cf.\ Theorem 10.8, \cite{Klebaner:2005}), process
\[
  W(t) = B(t)+\int_0^tH(u)du,
\]
is a $\Q_{\hat\beta}$-Brownian motion, and we have
\begin{align*}
X(t) &= \mu t + \diag\left(\hat \beta\right)\int_0^t \tau(u) du +\textup{diag}(\sigma) \int_0^t\diag\left(\tau^\sqt(u)\right)R^{\frac{1}{2}} dW(u) 
\end{align*}
As the following proposition states, $X \sim \textup{gNTS}_N(\alpha, \theta,\hat \beta, \mu, \sigma, R)$ is,  therefore, an NTS--process under measure $\Q_{\hat\beta}$. 

Let $\hat\theta=(\hat\theta_1,\hat\theta_2,\cdots, \hat\theta_N)^\tr \in \R_+$.
Using Proposition \ref{pro:Change of Measure Tau}, there is a measure $\Q_{\hat\theta_1,\hat\beta}$ equivalent to $\Q_{\hat\beta}$ under which $\Tau_1\sim\subTS(\alpha_1, 1, \hat\theta_1)$. Moreover, there is a measure $\Q_{\hat\theta_n, \hat\beta}$ equivalent to $\Q_{\hat\theta_{n-1},\hat\beta}$  under which $\Tau_n\sim\subTS(\alpha_n, 1, \hat\theta_n)$ for $n\in\{2,3,\cdots, N\}$. Therefore, $X \sim \textup{gNTS}_N(\alpha, \hat\theta, \hat\beta, \mu, \sigma, R)$ under the measure $\Q_{\hat\theta, \hat\beta}=\Q_{\hat\theta_{N},\hat\beta}$ . 
\end{proof}

\singlespacing
\bibliographystyle{decsci_mod}
\bibliography{refs_aaron_Quanto_2}

\end{document}